\documentclass[sigconf,9pt]{acmart}

\usepackage[english]{babel}
\usepackage{graphicx}
\usepackage{xcolor}

\definecolor{myred}{cmyk}{0.05, 0.32, 0.53, 0}
\definecolor{mygreen}{cmyk}{0.29, 0, 0.24, 0}
\definecolor{mydarkred}{cmyk}{0.22, 0.76, 0.100, 0.15}
\definecolor{mydarkblue}{cmyk}{0.85, 0.70, 0.62, 0.30}

\usepackage{tikz}
\usetikzlibrary{arrows}
\usetikzlibrary{positioning}

\usepackage{amsmath}

\usepackage{subcaption}

\copyrightyear{2024}
\acmYear{2024}
\setcopyright{rightsretained}
\acmConference[ACM SIGCOMM '24]{ACM SIGCOMM 2024 Conference}{August 4--8, 2024}{Sydney, NSW, Australia}
\acmBooktitle{ACM SIGCOMM 2024 Conference (ACM SIGCOMM '24), August 4--8, 2024, Sydney, NSW, Australia}
\acmDOI{10.1145/3651890.3672219}
\acmISBN{979-8-4007-0614-1/24/08}

\begin{document}

\title{Practical Rateless Set Reconciliation}

\author{Lei Yang}
\email{leiy@csail.mit.edu}
\orcid{0000-0003-3256-4700}
 \affiliation{%
   \institution{Massachusetts Institute of Technology}
    \streetaddress{77 Massachusetts Avenue}
    \city{Cambridge} 
    \state{MA} 
    \country{USA}
    \postcode{02139}
 }

\author{Yossi Gilad}
\email{yossigi@cs.huji.ac.il}
\orcid{0000-0003-3475-8322}
 \affiliation{%
   \institution{Hebrew University of Jerusalem}
    \city{Jerusalem} 
    \country{Israel}
 }

\author{Mohammad Alizadeh}
\email{alizadeh@csail.mit.edu}
\orcid{0000-0002-0014-6742}
 \affiliation{%
   \institution{Massachusetts Institute of Technology}
    \streetaddress{77 Massachusetts Avenue}
    \city{Cambridge} 
    \state{MA} 
    \country{USA}
    \postcode{02139}
 }

\begin{abstract}
Set reconciliation, where two parties hold fixed-length bit strings and run a protocol to learn the strings they are missing from each other, is a fundamental task in many distributed systems.
We present Rateless Invertible Bloom Lookup Tables (Rateless IBLTs), the first set reconciliation protocol, to the best of our knowledge, that achieves low computation cost and near-optimal communication cost across a wide range of scenarios: set differences of one to millions, bit strings of a few bytes to megabytes, and workloads injected by potential adversaries. Rateless IBLT is based on a novel encoder that incrementally encodes the set difference into an infinite stream of coded symbols, resembling rateless error-correcting codes. We compare Rateless IBLT with state-of-the-art set reconciliation schemes and demonstrate significant improvements. Rateless IBLT achieves $3$--$4\times$ lower communication cost than non-rateless schemes with similar computation cost, and $2$--$2000\times$ lower computation cost than schemes with similar communication cost.
We show the real-world benefits of Rateless IBLT by applying it to synchronize the state of the Ethereum blockchain, and demonstrate $5.6\times$ lower end-to-end completion time and $4.4\times$ lower communication cost compared to the system used in production.
\end{abstract}



\maketitle

\section{Introduction}

Recent years have seen growing interest in distributed applications, such as blockchains~\cite{bitcoin,ethereum,hyperledger}, social networks~\cite{mastodon}, mesh messaging~\cite{meshmessage}, and file hosting~\cite{ipfs}. 
In these applications, \emph{nodes} (participating servers) maintain replicas of the entire or part of the application state, and synchronize their replicas by exchanging messages on a peer-to-peer network.

The naive approach to solving this problem requires communication proportional to the size of the set one party holds. For example, some applications send the entire set while other applications have parties exchange the hashes of their items or a Bloom filter of their sets (the Bloom filter's size is proportional to the set size). A node then requests the items it is missing. All these solutions induce high overhead, especially when nodes have large, overlapping sets. This is a common scenario in distributed applications, such as nodes in a blockchain network synchronizing transactions or account balances, social media servers synchronizing users' posts, or a name system synchronizing certificates or revocation lists~\cite{gnusetrec}.

An emerging alternative solution to the state synchronization problem is \emph{set reconciliation}. 
It abstracts the application's state as a set and then uses a reconciliation protocol to synchronize replicas. Crucially, the overhead is determined by the \emph{set difference} size rather than the set size, allowing support for applications with very large states.
However, existing set reconciliation protocols suffer from at least one of two major caveats. First, most protocols are parameterized by the size of the set difference between the two participating parties. However, in practice, setting this parameter is difficult since scenarios such as temporal network disconnections or fluctuating load on the system make it challenging to know what the exact difference size will be ahead of time. Thus, application designers often resort to running online estimation protocols, which induce additional latency and only give a statistical estimate of the set difference size. Such estimators are inaccurate, forcing the application designers to tailor the parameters to the tail of the potential error, resulting in high communication overhead. The second type of caveat is that some set reconciliation protocols suffer from high decoding complexity, where the recipient has to run a quadratic-time or worse algorithm, with relatively expensive operations.

We propose a \emph{rateless} set reconciliation scheme called Rateless Invertible Bloom Lookup Tables (Rateless IBLT) that addresses these challenges. In Rateless IBLT, a sender generates an infinite stream of \emph{coded symbols} that encode the set difference, and the recipient can decode the set difference when they receive enough coded symbols. Rateless IBLT has no parameters and does not need an estimate of the set difference size. With overwhelming probability, the recipient can decode the set difference after receiving a number of coded symbols that are proportional to the set difference size rather than the entire set size, resulting in low overhead.
Rateless IBLT's coded symbols are {\em universal}. The same sequence of coded symbols can be used to reconcile any number of differences with any other set. Therefore, the sender can create coded symbols once and use them to synchronize with any number of peers.  
The latter property is particularly useful for applications such as blockchain peer-to-peer networks, where nodes may synchronize with multiple sources with overlapping states, since it allows the node to recover the union of their states using coded symbols it concurrently receives from all of them. 

In summary, we make the following contributions:
\begin{enumerate}
    \item The design of Rateless IBLT, the first set reconciliation protocol that achieves low computation cost and near-optimal communication cost across a wide range of scenarios: set differences of one to millions, bit strings of a few bytes to megabytes, and workloads injected by potential adversaries.

    \item A mathematical analysis of Rateless IBLT's communication and computation costs. We prove that when the set difference size $d$ goes to infinity, Rateless IBLT reconciles $d$ differences with $1.35d$ communication. We show in simulations that the communication cost is between $1.35d$ to $1.72d$ on average for all values of $d$ and that it quickly converges to $1.35\times$ when $d$ is in the low hundreds. 

    \item An implementation of Rateless IBLT as a library. When reconciling $1000$ differences, our implementation can process input data (sets being reconciled) at $120$ MB/s using a \emph{single} core of a 2016-model CPU.
    
    \item Extensive experiments comparing Rateless IBLT with state-of-the-art solutions. Rateless IBLT achieves $3$--$4\times$ lower communication cost than regular IBLT~\cite{iblt} and MET-IBLT~\cite{rate-compatible}, two non-rateless schemes; and $2$--$2000\times$ lower computation cost than PinSketch~\cite{pinsketch}.

    \item Demonstration of Rateless IBLT's real-world benefits by applying our implementation to synchronize the account states of the Ethereum blockchain. Compared to Merkle trie~\cite{merkle}, today's de facto solution, Rateless IBLT achieves $5.6\times$ lower completion time and $4.4\times$ lower communication cost on historic traces.
\end{enumerate}

\section{Motivation and Related Work}
\label{sec:related}

We first formally define the set reconciliation problem~\cite{cpisync,difference}. Let $A$ and $B$ be two sets containing items (bit strings) of the same length $\ell$. $A$ and $B$ are stored by two distinct parties, Alice and Bob. They want to efficiently compute the symmetric difference of $A$ and $B$, i.e., $(A \cup B) \setminus (A \cap B)$, denoted as $A \bigtriangleup B$. By convention~\cite{difference}, we assume that only one of the parties, Bob, wants to compute $A \bigtriangleup B$ because he can send the result to Alice afterward if needed.

While straightforward solutions exist, such as exchanging Bloom filters~\cite{bloomfilter} or hashes of the items, they incur $O(|A| + |B|)$ communication and computation costs. The costs can be improved to logarithmic by hashing the sets into Merkle tries~\cite{merkle}, where a trie node on depth $i$ is the hash of a $\left(1/2^i\right)$-fraction of the set. Alice and Bob traverse and compare their tries, only descending into a sub-trie (subset) if their roots (hashes) differ. However, the costs are still dependent on $|A|, |B|$, and now takes $O(\log|A| + \log|B|)$ round trips.

In contrast, the information-theoretic lower bound~\cite[\S~2]{cpisync} of the communication cost is $d\ell$, where $d = |A \bigtriangleup B|$.\footnote{More precisely, the lower bound is $d\ell - d\log_2 d$~\cite[\S~2]{cpisync}, but the second term can be neglected when $d \ll 2^\ell$.} State-of-the-art solutions get close to this lower bound using techniques from coding theory.
On a high level, we can view $B$ as a copy of $A$ with $d$ errors (insertions and/or deletions), and the goal of set reconciliation is to correct these errors.
Alice \emph{encodes} $A$ into a list of $m$ \emph{coded symbols} and sends them to Bob. Bob then uses the coded symbols and $B$ to decode $A \bigtriangleup B$.
The coded symbols are the parity data in a systematic error-correcting code that can correct set insertions and deletions~\cite{reduction}. Using appropriate codes, it takes $m=O(d)$ coded symbols, each of length $O(\ell)$, to correct the $d$ errors, resulting in a communication cost of $O(d\ell)$.

The performance of existing solutions varies depending on the codes they use.
Characteristic Polynomial Interpolation (CPI)~\cite{cpisync} uses a variant of Reed-Solomon codes~\cite{reedsolomon}, where coded symbols are evaluations of a polynomial uniquely constructed from $A$. CPI has a communication cost of $d\ell$, achieving the information-theoretic lower bound. However, its computation cost is $O(|A|d\ell)$ for Alice, and $O(|B|d\ell + d^3 \ell^4)$ for Bob. The latter was improved to $O(|B|d\ell + d^2 \ell^2)$ in PinSketch~\cite{pinsketch,minisketch} using BCH codes~\cite{bchcodes} that are easier to decode. Nevertheless, as we show in \S~\ref{sec:bench-compute}, computation on both Alice and Bob quickly becomes intractable even at moderate $|A|$, $|B|$, and $d$, limiting its applicability. For example, Shrec~\cite{shrec} attempted to use PinSketch to synchronize transactions in a high-throughput blockchain but found that its high computation complexity severely limits system throughput~\cite[\S~5.2]{shrec}.

Invertible Bloom Lookup Tables (IBLTs)~\cite{iblt} use sparse graph codes similar to LT~\cite{lt} and LDPC~\cite{ldpc} codes. Each set item is summed into $k$ coded symbols, denoted as its neighbors in a random, sparse graph. Some variants also consider graphs with richer structures such as varying $k$ depending on the set item~\cite{irregular}. The computation cost is $O(|A|k\ell)$ for Alice, and $O((|B|+d)k\ell)$ for Bob. The communication cost is $O(d\ell)$ with a coefficient strictly larger than $1$ (e.g., $4$--$10$ for small $d$, see \S~\ref{sec:bench-communication}). Due to their random nature, IBLTs
may fail to decode even if properly parameterized~\cite{difference}.
We provide more background on IBLTs in \S~\ref{sec:background}.

The aforementioned discussions assume that the codes are properly \emph{parameterized}. In particular, we need to decide $m$, the number of coded symbols Alice sends to Bob. Decoding will fail if $m$ is too small compared to $d$, and we incur redundant communication and computation if $m$ is too large. The optimal choice of $m$ is thus a function of $d$. However, accurate prediction of $d$ is usually difficult~\cite{erlay,shrec}, and sometimes outright impossible~\cite{graphene}. 
Existing systems often resort to online estimation protocols~\cite{difference} and over-provision $m$ to accommodate the ensuing errors~\cite{difference,graphene}.

\smallskip\noindent\textbf{The case for rateless reconciliation.} A key feature of Rateless IBLT is that it can generate an infinite stream of coded symbols for a set, resembling rateless error-correcting codes~\cite{fountain}. For any $m>0$, the first $m$ coded symbols can reconcile $O(m)$ set differences with a coefficient close to $1$ ($0.74$ in most cases, see \S~\ref{sec:bench-communication}).
This means that Rateless IBLT does not require parameterization, making real-world deployments easy and robust. Alice simply keeps sending coded symbols to Bob, and Bob can decode as soon as he receives enough---which we show analytically (\S~\ref{sec:analysis}) and experimentally (\S~\ref{sec:bench-communication}) to be about $1.35d$ in most cases---coded symbols. Neither Alice nor Bob needs to know $d$ beforehand. The encoding and the decoding algorithms have zero parameters.

The concept of incrementally generating coded symbols is first mentioned in CPI~\cite{cpisync}. However, as mentioned before, its real-world use has been limited due to the high computation cost. We discuss these limitations in \S~\ref{sec:bench}, and demonstrate that Rateless IBLT reduces the computation cost by $2$--$2000\times$, while incurring a communication cost of less than $2\times$ the information-theoretic lower bound.
Concurrently with our work, MET-IBLT~\cite{rate-compatible} proposes to simultaneously optimize the parameters of IBLTs for multiple pre-selected values of $d$, e.g., $d_1, d_2, \dots, d_n$, such that a list of coded symbols for $d_i$ is a prefix/suffix of that for $d_j$. However, it only considers a few values of $d$ due to the complexity of the optimization so still requires workload-dependent parameterization. As we show in \S~\ref{sec:bench-communication}, its communication cost is $4$--$10\times$ higher for the $d$ values that are not optimized for. In addition, MET-IBLT does not provide a practical algorithm to incrementally generate coded symbols. Rateless IBLT does not have any of these issues.

Rateless IBLT offers additional benefits. Imagine that Alice has the canonical system state, and multiple peers wish to reconcile their states with Alice. In a non-rateless scheme, Alice must separately produce coded symbols for each peer depending on the particular number of differences $d$. This incurs additional computation and storage I/Os for every peer she reconciles with. And, more importantly, Alice must produce the coded symbols on the fly because she does not know $d$ before a peer arrives. In comparison, using Rateless IBLT, Alice simply maintains a universal sequence of coded symbols and streams it to anyone who wishes to reconcile. Rateless IBLT also allows her to incrementally update the coded symbols as she modifies the state (set), further amortizing the encoding costs.

To the best of our knowledge, Rateless IBLT is the first set reconciliation solution that simultaneously achieves the following properties:
\begin{itemize}
    \item \textbf{Ratelessness.} The encoder generates an infinite sequence of coded symbols, capable of reconciling any number of differences $d$ with low overhead.
    \item \textbf{Universality.} The same algorithm works efficiently for any $|A|$, $|B|$, $d$, and $\ell$ without any parameter.
    \item \textbf{Low communication cost.} The average communication cost peaks at $1.72d\ell$ when $d=4$, and quickly converges to $1.35d\ell$ when $d$ is at the low hundreds.
    \item \textbf{Low computation cost.} Encoding costs $O(\ell \log d)$ per set item, and decoding costs $O(\ell \log d)$ per difference. In practice, a single core on a 2016-model CPU can encode (decode) 3.4 million items (differences) per second when $d=1000$ and $\ell=8$ bytes.
\end{itemize}
We demonstrate these advantages by comparing with all the aforementioned schemes in \S~\ref{sec:bench} and applying Rateless IBLT to a potential application in \S~\ref{sec:eval}.

\section{Background}
\label{sec:background}

Our Rateless IBLT retains the format of coded symbols and the decoding algorithm of regular IBLTs, but employs a new encoder that is oblivious to the number of differences to reconcile. In this section, we provide the necessary background on IBLTs~\cite{difference,iblt} and explain why regular IBLTs fail to provide the rateless property that we desire.  We discuss the new rateless encoder in the next section. 

On a high level, an IBLT is an encoding of a set. We call the items (bit
strings) in the set the \emph{source symbols}, and an IBLT comprises a list of
$m$ \emph{coded symbols}.
Each source symbol is \emph{mapped} to $k$ coded symbols uniformly at random, e.g., by using $k$ hash functions.
Here, $m$ and $k$ are design parameters.

\smallskip\noindent
\textbf{Coded symbol format.}
A coded symbol contains two fields: \texttt{sum}, the bitwise exclusive-or (XOR) sum of the source symbols mapped to it; and \texttt{checksum}, the bitwise XOR sum of the hashes of the source symbols mapped to it. In practice, there is usually a third field, \texttt{count}, which we will discuss shortly. Fig.~\ref{fig:iblt-example} provides an example.

\begin{figure}
    \centering
    \begin{tikzpicture}[>=latex']
        \small
 
        \node[anchor=east] (setlabel) at (-0.6,-0.4) {$A$};
        \draw [draw=gray, thick, dashed, rounded corners] (-0.5,0.1) rectangle (3.2,-0.9);

        
        \node[draw, circle, text centered,fill=myred] (x0) at (0, -0.4) {$x_0$};
        \node[draw, circle, text centered,fill=myred] (x1) at (0.9, -0.4) {$x_1$};
        \node[draw, circle, text centered,fill=myred] (x2) at (1.8, -0.4) {$x_2$};
        \node[draw, circle, text centered,fill=myred] (x3) at (2.7, -0.4) {$x_3$};

        \node[anchor=east] (codelabel) at (-0.6,-2) {$\text{IBLT}(A)$};
        \draw [draw=gray, thick, dashed, rounded corners] (-0.5,-1.5) rectangle (5.0,-2.5);
        \node[draw, rectangle, text centered, minimum size=1.8em,fill=mygreen] (a0) at (0, -2) {$a_0$};
        \node[draw, rectangle, text centered, minimum size=1.8em,fill=mygreen] (a1) at (0.9, -2) {$a_1$};
        \node[draw, rectangle, text centered, minimum size=1.8em,fill=mygreen] (a2) at (1.8, -2) {$a_2$};
        \node[draw, rectangle, text centered, minimum size=1.8em,fill=mygreen] (a3) at (2.7, -2) {$a_3$};
        \node[draw, rectangle, text centered, minimum size=1.8em,fill=mygreen] (a4) at (3.6, -2) {$a_4$};
        \node[draw, rectangle, text centered, minimum size=1.8em,fill=mygreen] (a5) at (4.5, -2) {$a_5$};

        \draw [] (x0) -- node[] {} (a0);
        \draw [] (x0) -- node[] {} (a1);
        \draw [] (x0) -- node[] {} (a2);
        \draw [] (x1) -- node[] {} (a0);
        \draw [] (x1) -- node[] {} (a3);
        \draw [] (x1) -- node[] {} (a4);
        \draw [] (x2) -- node[] {} (a1);
        \draw [] (x2) -- node[] {} (a2);
        \draw [] (x2) -- node[] {} (a3);
        \draw [] (x3) -- node[] {} (a2);
        \draw [] (x3) -- node[] {} (a4);
        \draw [] (x3) -- node[] {} (a5);

    \end{tikzpicture}
\caption{\label{fig:iblt-example} Example of constructing a regular IBLT for set $A$ with source symbols $x_0, x_1, x_2, x_3$. The IBLT has $m=6$ coded symbols: $a_0, a_1, a_2, a_3, a_4, a_5$. Each source symbol is mapped to $k=3$ coded symbols. Solid lines represent the mapping between source and coded symbols. For example, for $a_4$, $\mathtt{sum}=x_1\oplus x_3$, $\mathtt{checksum}=\mathtt{Hash}(x_1)\oplus \mathtt{Hash}(x_3)$, and $\mathtt{count}=2$. $\oplus$ is the bitwise exclusive-or operator.}
\end{figure}
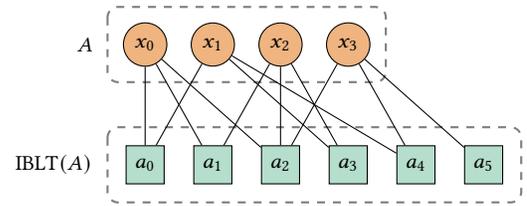

\smallskip\noindent
\textbf{Peeling decoder.} To recover source symbols from a list of coded symbols, the decoder runs a recursive procedure called ``peeling''.
We say a coded symbol is \emph{pure} when exactly one source symbol is mapped to it; or, equivalently, when its \texttt{checksum} equals the hash of its \texttt{sum}~\cite{difference}.\footnote{Unless there is a hash collision, which happens with negligible probability in the length of the hash. See \S~\ref{sec:hash-function}.}
In this case, its \texttt{sum} field is the source symbol itself, which is now recovered. 
The decoder then removes the recovered source symbol from any other coded symbols it is mapped to
(determined by the $k$ agreed-upon hash functions),
by XOR-ing the source symbol and its hash into their \texttt{sum} and \texttt{checksum} fields, respectively. This process may generate additional pure symbols; the decoder repeats until no pure symbols are left. Decoding fails if it stops before recovering all source symbols~\cite{iblt}. Fig.~\ref{fig:iblt-decode} shows the example of decoding the IBLT in Fig.~\ref{fig:iblt-example}. 

\begin{figure}
    \centering%
\begin{subfigure}[t]{0.33\columnwidth}
\centering%
\begin{tikzpicture}[>=latex']
\small
\node[draw, circle, fill=myred, minimum size=0.8em, inner sep=0pt] (x0) at (0, -0.4) {};
\node[draw, circle, fill=myred, minimum size=0.8em, inner sep=0pt] (x1) at (0.4, -0.4) {};
\node[draw, circle, fill=myred, minimum size=0.8em, inner sep=0pt] (x2) at (0.8, -0.4) {};
\node[draw, circle, fill=mydarkred, minimum size=0.8em, inner sep=0pt] (x3) at (1.2, -0.4) {};

        \node[draw, rectangle, minimum size=0.72em,fill=mygreen] (a0) at (0, -1.5) {};
        \node[draw, rectangle, minimum size=0.72em,fill=mygreen] (a1) at (0.4, -1.5) {};
        \node[draw, rectangle, minimum size=0.72em,fill=mygreen] (a2) at (0.8, -1.5) {};
        \node[draw, rectangle, minimum size=0.72em,fill=mygreen] (a3) at (1.2, -1.5) {};
        \node[draw, rectangle, minimum size=0.72em,fill=mygreen] (a4) at (1.6, -1.5) {};
        \node[draw, rectangle, minimum size=0.72em,fill=mydarkblue] (a5) at (2, -1.5) {};

        \draw [] (x0) -- node[] {} (a0);
        \draw [] (x0) -- node[] {} (a1);
        \draw [] (x0) -- node[] {} (a2);
        \draw [] (x1) -- node[] {} (a0);
        \draw [] (x1) -- node[] {} (a3);
        \draw [] (x1) -- node[] {} (a4);
        \draw [] (x2) -- node[] {} (a1);
        \draw [] (x2) -- node[] {} (a2);
        \draw [] (x2) -- node[] {} (a3);
        \draw [dashed] (x3) -- node[] {} (a2);
        \draw [dashed] (x3) -- node[] {} (a4);
        \draw [dashed] (x3) -- node[] {} (a5);
    \end{tikzpicture}
    \caption{Iteration 1}
     \end{subfigure}\hfill\begin{subfigure}[t]{0.33\columnwidth}\centering%
    \begin{tikzpicture}[>=latex']
    \small
        \node[draw, circle, fill=myred, minimum size=0.8em, inner sep=0pt] (x0) at (0, -0.4) {};
\node[draw, circle, fill=mydarkred, minimum size=0.8em, inner sep=0pt] (x1) at (0.4, -0.4) {};
\node[draw, circle, fill=myred, minimum size=0.8em, inner sep=0pt] (x2) at (0.8, -0.4) {};
\node[draw, circle, fill=mydarkred, minimum size=0.8em, inner sep=0pt] (x3) at (1.2, -0.4) {};

        \node[draw, rectangle, minimum size=0.72em,fill=mygreen] (a0) at (0, -1.5) {};
        \node[draw, rectangle, minimum size=0.72em,fill=mygreen] (a1) at (0.4, -1.5) {};
        \node[draw, rectangle, minimum size=0.72em,fill=mygreen] (a2) at (0.8, -1.5) {};
        \node[draw, rectangle, minimum size=0.72em,fill=mygreen] (a3) at (1.2, -1.5) {};
        \node[draw, rectangle, minimum size=0.72em,fill=mydarkblue] (a4) at (1.6, -1.5) {};

        \draw [] (x0) -- node[] {} (a0);
        \draw [] (x0) -- node[] {} (a1);
        \draw [] (x0) -- node[] {} (a2);
        \draw [dashed] (x1) -- node[] {} (a0);
        \draw [dashed] (x1) -- node[] {} (a3);
        \draw [dashed] (x1) -- node[] {} (a4);
        \draw [] (x2) -- node[] {} (a1);
        \draw [] (x2) -- node[] {} (a2);
        \draw [] (x2) -- node[] {} (a3);
    \end{tikzpicture}
         \caption{Iteration 2}
     \end{subfigure}\hfill\begin{subfigure}[t]{0.33\columnwidth}\centering%
    \begin{tikzpicture}[>=latex']
    \small
        \node[draw, circle, fill=mydarkred, minimum size=0.8em, inner sep=0pt] (x0) at (0, -0.4) {};
\node[draw, circle, fill=mydarkred, minimum size=0.8em, inner sep=0pt] (x1) at (0.4, -0.4) {};
\node[draw, circle, fill=mydarkred, minimum size=0.8em, inner sep=0pt] (x2) at (0.8, -0.4) {};
\node[draw, circle, fill=mydarkred, minimum size=0.8em, inner sep=0pt] (x3) at (1.2, -0.4) {};

        \node[draw, rectangle, minimum size=0.72em,fill=mydarkblue] (a0) at (0, -1.5) {};
        \node[draw, rectangle, minimum size=0.72em,fill=mygreen] (a1) at (0.4, -1.5) {};
        \node[draw, rectangle, minimum size=0.72em,fill=mygreen] (a2) at (0.8, -1.5) {};
        \node[draw, rectangle, minimum size=0.72em,fill=mydarkblue] (a3) at (1.2, -1.5) {};

        \draw [dashed] (x0) -- node[] {} (a0);
        \draw [dashed] (x0) -- node[] {} (a1);
        \draw [dashed] (x0) -- node[] {} (a2);
        \draw [dashed] (x2) -- node[] {} (a1);
        \draw [dashed] (x2) -- node[] {} (a2);
        \draw [dashed] (x2) -- node[] {} (a3);
    \end{tikzpicture}
         \caption{Iteration 3}
     \end{subfigure}
\caption{\label{fig:iblt-decode} Example of decoding the IBLT in Fig.~\ref{fig:iblt-example} using peeling. Dark colors represent pure coded symbols at the beginning of each iteration, and source symbols recovered so far. Dashed edges are removed at the end of each iteration, by XOR-ing the source symbol (now recovered) and its hash on one end of the edge into the \texttt{sum} and \texttt{checksum} fields of the coded symbol on the other end.}
\end{figure}
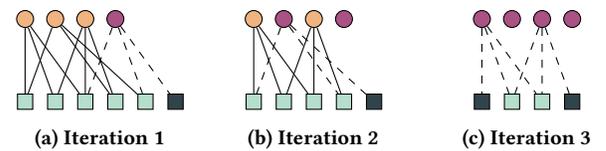

\smallskip\noindent
\textbf{Subtraction of coded symbols.} $a \oplus b$ denotes subtraction of two coded symbols $a, b$. For the resulting coded symbol, its $\texttt{sum}$ is the bitwise XOR of $a\mathtt{.sum}$ and $b\mathtt{.sum}$; its $\texttt{checksum}$ is the bitwise XOR of $a\mathtt{.checksum}$ and $b\mathtt{.checksum}$; and its $\texttt{count}$ is $a\mathtt{.count}-b\mathtt{.count}$.

\smallskip\noindent
\textbf{Set reconciliation using IBLTs.} 
IBLTs with the same parameter configuration ($m$, $k$, and hash functions mapping from source symbols to coded symbols) can be subtracted~\cite{difference}.
For any two sets $A$ and $B$, $\text{IBLT}(A) \oplus \text{IBLT}(B) = \text{IBLT}(A\bigtriangleup B)$, where the $\oplus$ operator subtracts each corresponding pair of coded symbols from the two IBLTs. This is because if a source symbol is present in both $A$ and $B$, then it is XOR-ed twice into each coded symbol it is mapped to in $\text{IBLT}(A) \oplus \text{IBLT}(B)$, resulting in no effect. As a result, $\text{IBLT}(A) \oplus \text{IBLT}(B)$ is an encoding of the source symbols that appear exactly once across $A$ and $B$, i.e., $A \bigtriangleup B$.

To reconcile $A$ and $B$, Alice sends $\text{IBLT}(A)$ to Bob, who then computes and decodes $\text{IBLT}(A) \oplus \text{IBLT}(B)$ to recover $A \bigtriangleup B$. To determine whether a recovered source symbol belongs to $A$ or $B$, we use the \texttt{count} field.\footnote{Alternatively, Bob may try looking up each item in $A\bigtriangleup B$ against $B$, but this requires indexing $B$, which is undesirable when $|B|$ is large.} It records the number of source symbols mapped to a coded symbol. When a coded symbol is pure, $\texttt{count} =1 $ indicates that the recovered source symbol is exclusive to $A$, and $\texttt{count}=-1$ indicates $B$~\cite{difference}.

\smallskip\noindent
\textbf{Limitations of IBLTs.} IBLTs are not rateless. An IBLT with a particular set of parameters $m, k$ only works for a narrow range of difference size $d$. It quickly becomes inefficient to use it for more or fewer differences than parameterized for.
In \S~\ref{appendix:iblt-musing}, we show Theorems~\ref{thm:iblt-undersize} and~\ref{thm:iblt-oversize}, which we summarize informally here.
First,
with high probability, Bob cannot recover \emph{any} source symbol in $A \bigtriangleup B$ when $d>m$, making undersized IBLTs completely useless.
On the other hand, we cannot simply default to a very large $m$ to accommodate a potentially large $d$. If $d$ turns out to be small, i.e., $d \ll m$, Alice still has to send almost the entire IBLT ($m$ coded symbols) for Bob to decode successfully, leading to high communication cost. Alice cannot dynamically enlarge $m$, either. Each source symbol is already uniformly mapped to $k$ out of $m$ coded symbols upon encoding. Increasing $m$ post hoc would require remapping the source symbols to the expanded space of coded symbols so that the mapping remains uniform. This requires Alice to rebuild and re-send the entire IBLT. Figs.~\ref{fig:sample-traditional-small},~\ref{fig:sample-traditional-big} show an example.

\section{Design}
\label{sec:design}

\begin{figure}
    \centering%
\begin{subfigure}[t]{0.5\columnwidth}%
\centering%
\begin{tikzpicture}[>=latex']
\small
\node[anchor=west,inner sep=0,outer sep=0] (sourcelabel) at (-1.81,-0.4) {\small$A\bigtriangleup B$};
\node[anchor=west,inner sep=0,outer sep=0] (codedlabel) at (-1.81,-1.5) {\small IBLT};
\node[draw, circle, fill=myred, minimum size=0.8em, inner sep=0pt] (x0) at (0, -0.4) {};
\node[draw, circle, fill=myred, minimum size=0.8em, inner sep=0pt] (x1) at (0.4, -0.4) {};
\node[draw, circle, fill=myred, minimum size=0.8em, inner sep=0pt] (x2) at (0.8, -0.4) {};
\node[draw, circle, fill=myred, minimum size=0.8em, inner sep=0pt] (x3) at (1.2, -0.4) {};
\node[draw, circle, fill=myred, minimum size=0.8em, inner sep=0pt] (x4) at (1.6, -0.4) {};
\node[draw, rectangle, minimum size=0.72em,fill=mygreen] (a0) at (0, -1.5) {};
\node[draw, rectangle, minimum size=0.72em,fill=mygreen] (a1) at (0.4, -1.5) {};
\node[draw, rectangle, minimum size=0.72em,fill=mygreen] (a2) at (0.8, -1.5) {};
\node[draw, rectangle, minimum size=0.72em,fill=mygreen] (a3) at (1.2, -1.5) {};

\draw [color=mydarkblue] (x0) -- node[] {} (a1);
\draw [color=mydarkblue] (x1) -- node[] {} (a1);
\draw [color=mydarkblue] (x2) -- node[] {} (a1);
\draw [color=mydarkred] (x3) -- node[] {} (a1);
\draw [color=mydarkblue] (x4) -- node[] {} (a1);
\draw [color=mydarkblue] (x0) -- node[] {} (a0);
\draw [color=mydarkblue] (x3) -- node[] {} (a0);
\draw [color=mydarkblue] (x4) -- node[] {} (a0);
\draw [color=mydarkblue] (x1) -- node[] {} (a2);
\draw [color=mydarkred] (x2) -- node[] {} (a2);
\draw [color=mydarkblue] (x4) -- node[] {} (a2);
\draw [color=mydarkred] (x0) -- node[] {} (a3);
\draw [color=mydarkblue] (x2) -- node[] {} (a3);
\draw [color=mydarkblue] (x3) -- node[] {} (a3);
\draw [color=mydarkblue] (x1) -- node[] {} (a3);
\end{tikzpicture}
\caption{Regular IBLT, $m=4$\label{fig:sample-traditional-small}}
\end{subfigure}\hfill\begin{subfigure}[t]{0.5\columnwidth}\centering%
\begin{tikzpicture}[>=latex']
\small
\node[draw, circle, fill=myred, minimum size=0.8em, inner sep=0pt] (x0) at (0, -0.4) {};
\node[draw, circle, fill=myred, minimum size=0.8em, inner sep=0pt] (x1) at (0.4, -0.4) {};
\node[draw, circle, fill=myred, minimum size=0.8em, inner sep=0pt] (x2) at (0.8, -0.4) {};
\node[draw, circle, fill=myred, minimum size=0.8em, inner sep=0pt] (x3) at (1.2, -0.4) {};
\node[draw, circle, fill=myred, minimum size=0.8em, inner sep=0pt] (x4) at (1.6, -0.4) {};
\node[draw, rectangle, minimum size=0.72em,fill=mydarkblue] (a0) at (0, -1.5) {};
\node[draw, rectangle, minimum size=0.72em,fill=mydarkblue] (a1) at (0.4, -1.5) {};
\node[draw, rectangle, minimum size=0.72em,fill=mydarkblue] (a2) at (0.8, -1.5) {};
\node[draw, rectangle, minimum size=0.72em,fill=mydarkblue] (a3) at (1.2, -1.5) {};
\node[draw, rectangle, minimum size=0.72em,fill=mydarkblue] (a4) at (1.6, -1.5) {};
\node[draw, rectangle, minimum size=0.72em,fill=mydarkblue] (a5) at (2, -1.5) {};
\node[draw, rectangle, minimum size=0.72em,fill=mydarkblue] (a6) at (2.4, -1.5) {};

\draw [color=mydarkblue] (x1) -- node[] {} (a0);
\draw [color=mydarkblue] (x2) -- node[] {} (a0);
\draw [color=mydarkblue] (x4) -- node[] {} (a0);
\draw [color=mydarkred] (x3) -- node[] {} (a1);
\draw [color=mydarkblue] (x4) -- node[] {} (a1);
\draw [color=mydarkred] (x2) -- node[] {} (a2);
\draw [color=mydarkred] (x0) -- node[] {} (a3);
\draw [color=mydarkblue] (x4) -- node[] {} (a3);
\draw [color=mydarkblue] (x0) -- node[] {} (a4);
\draw [color=mydarkblue] (x2) -- node[] {} (a4);
\draw [color=mydarkblue] (x0) -- node[] {} (a5);
\draw [color=mydarkblue] (x1) -- node[] {} (a5);
\draw [color=mydarkblue] (x3) -- node[] {} (a5);
\draw [color=mydarkblue] (x1) -- node[] {} (a6);
\draw [color=mydarkblue] (x3) -- node[] {} (a6);
\end{tikzpicture}
\caption{Regular IBLT, $m=7$\label{fig:sample-traditional-big}}
\end{subfigure}
\par\bigskip
\begin{subfigure}[t]{0.5\columnwidth}%
\centering%
\begin{tikzpicture}[>=latex']
\small
\node[anchor=west,inner sep=0,outer sep=0] (sourcelabel) at (-1.81,-0.4) {\small $A\bigtriangleup B$};
\node[anchor=west,inner sep=0,outer sep=0] (codedlabel) at (-1.81,-1.5) {\small Rateless IBLT};
\node[draw, circle, fill=myred, minimum size=0.8em, inner sep=0pt] (x0) at (0, -0.4) {};
\node[draw, circle, fill=myred, minimum size=0.8em, inner sep=0pt] (x1) at (0.4, -0.4) {};
\node[draw, circle, fill=myred, minimum size=0.8em, inner sep=0pt] (x2) at (0.8, -0.4) {};
\node[draw, circle, fill=myred, minimum size=0.8em, inner sep=0pt] (x3) at (1.2, -0.4) {};
\node[draw, circle, fill=myred, minimum size=0.8em, inner sep=0pt] (x4) at (1.6, -0.4) {};
\node[draw, rectangle, minimum size=0.72em,fill=mygreen] (a0) at (0, -1.5) {};
\node[draw, rectangle, minimum size=0.72em,fill=mygreen] (a1) at (0.4, -1.5) {};
\node[draw, rectangle, minimum size=0.72em,fill=mygreen] (a2) at (0.8, -1.5) {};
\node[draw, rectangle, minimum size=0.72em,fill=mygreen] (a3) at (1.2, -1.5) {};

\draw [color=mydarkred] (x0) -- node[] {} (a0);
\draw [color=mydarkred] (x1) -- node[] {} (a0);
\draw [color=mydarkred] (x2) -- node[] {} (a0);
\draw [color=mydarkred] (x3) -- node[] {} (a0);
\draw [color=mydarkred] (x4) -- node[] {} (a0);
\draw [color=mydarkred] (x0) -- node[] {} (a1);
\draw [color=mydarkred] (x2) -- node[] {} (a1);
\draw [color=mydarkred] (x3) -- node[] {} (a1);
\draw [color=mydarkred] (x4) -- node[] {} (a1);
\draw [color=mydarkred] (x2) -- node[] {} (a2);
\draw [color=mydarkred] (x3) -- node[] {} (a2);
\draw [color=mydarkred] (x2) -- node[] {} (a3);
\draw [color=mydarkred] (x3) -- node[] {} (a3);
\draw [color=mydarkred] (x4) -- node[] {} (a3);
\end{tikzpicture}
\caption{Rateless IBLT, prefix of $m=4$\label{fig:sample-riblt-small}}
\end{subfigure}\hfill\begin{subfigure}[t]{0.5\columnwidth}\centering%
\begin{tikzpicture}[>=latex']
\small
\node[draw, circle, fill=myred, minimum size=0.8em, inner sep=0pt] (x0) at (0, -0.4) {};
\node[draw, circle, fill=myred, minimum size=0.8em, inner sep=0pt] (x1) at (0.4, -0.4) {};
\node[draw, circle, fill=myred, minimum size=0.8em, inner sep=0pt] (x2) at (0.8, -0.4) {};
\node[draw, circle, fill=myred, minimum size=0.8em, inner sep=0pt] (x3) at (1.2, -0.4) {};
\node[draw, circle, fill=myred, minimum size=0.8em, inner sep=0pt] (x4) at (1.6, -0.4) {};
\node[draw, rectangle, minimum size=0.72em,fill=mygreen] (a0) at (0, -1.5) {};
\node[draw, rectangle, minimum size=0.72em,fill=mygreen] (a1) at (0.4, -1.5) {};
\node[draw, rectangle, minimum size=0.72em,fill=mygreen] (a2) at (0.8, -1.5) {};
\node[draw, rectangle, minimum size=0.72em,fill=mygreen] (a3) at (1.2, -1.5) {};
\node[draw, rectangle, minimum size=0.72em,fill=mydarkblue] (a4) at (1.6, -1.5) {};
\node[draw, rectangle, minimum size=0.72em,fill=mydarkblue] (a5) at (2, -1.5) {};
\node[draw, rectangle, minimum size=0.72em,fill=mydarkblue] (a6) at (2.4, -1.5) {};

\draw [color=mydarkred] (x0) -- node[] {} (a0);
\draw [color=mydarkred] (x1) -- node[] {} (a0);
\draw [color=mydarkred] (x2) -- node[] {} (a0);
\draw [color=mydarkred] (x3) -- node[] {} (a0);
\draw [color=mydarkred] (x4) -- node[] {} (a0);
\draw [color=mydarkred] (x0) -- node[] {} (a1);
\draw [color=mydarkred] (x2) -- node[] {} (a1);
\draw [color=mydarkred] (x3) -- node[] {} (a1);
\draw [color=mydarkred] (x4) -- node[] {} (a1);
\draw [color=mydarkred] (x2) -- node[] {} (a2);
\draw [color=mydarkred] (x3) -- node[] {} (a2);
\draw [color=mydarkred] (x2) -- node[] {} (a3);
\draw [color=mydarkred] (x3) -- node[] {} (a3);
\draw [color=mydarkred] (x4) -- node[] {} (a3);
\draw [color=mydarkblue] (x0) -- node[] {} (a4);
\draw [color=mydarkblue] (x2) -- node[] {} (a5);
\draw [color=mydarkblue] (x3) -- node[] {} (a5);
\draw [color=mydarkblue] (x2) -- node[] {} (a6);
\end{tikzpicture}
\caption{Rateless IBLT, prefix of $m=7$\label{fig:sample-riblt-big}}
\end{subfigure}
\caption{\label{fig:iblt-prefix}Regular IBLTs and prefixes of Rateless IBLT for $5$ source symbols. Figs. a, c (left) have too few coded symbols and are undecodable. Figs. b, d (right) are decodable. Red edges are common across each row. Dark coded symbols in Figs. b, d are new or changed compared to their counterparts in Figs. a, c. Imagine that Alice sends $4$ coded symbols but Bob fails to decode. In regular IBLT, in order to enlarge $m$, she has to send all $7$ coded symbols since the existing $4$ symbols also changed. In Rateless IBLT, she only needs to send the $3$ new symbols. The existing $4$ symbols stay the same.}
\end{figure}
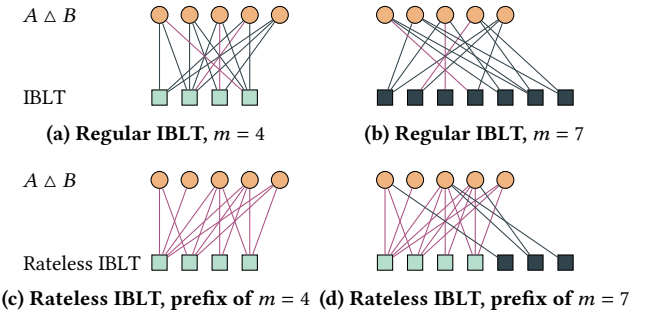

For any set $S$, Rateless IBLT defines an infinite sequence of coded symbols. Intuitively, an infinite number of coded symbols for $A$ and $B$ allows Rateless IBLT to accommodate an arbitrarily large set difference.
Every \emph{prefix} of this infinite sequence functions like a normal IBLT and can reconcile a number of differences proportional to its length. 
Meanwhile, because these prefixes belong to a common infinite sequence, 
Alice simply streams the sequence until Bob receives a long enough prefix to decode.
For any $d > 0$, on average, reconciling $d$ differences requires only the first $1.35d$--$1.72d$ coded symbols in the sequence.

\subsection{Coded Symbol Sequence}
\label{sec:sequence}

Our main task is to design the algorithm that encodes any set $S$ into an infinite sequence of coded symbols, denoted as $s_0, s_1, s_2, \dots$. It should provide the following properties:
\begin{itemize}
    \item \textbf{Decodability.} With high probability, the peeling decoder can recover all source symbols in a set $S$ using a prefix of $s_0, s_1, s_2, \dots$ with length~$O(|S|)$.
    \item \textbf{Linearity.} For any sets $A$ and $B$, $a_0 \oplus b_0, a_1 \oplus b_1, a_2 \oplus b_2, \dots$ is the coded symbol sequence for $A \bigtriangleup B$. 
    \item \textbf{Universality.} The encoding algorithm does not need any extra information other than the set being encoded.
\end{itemize}

These properties allow us to build the following simple protocol for rateless reconciliation.
To reconcile $A$ and $B$, Alice incrementally sends $a_0, a_1, a_2, \dots$ to Bob. Bob computes $a_0\oplus b_0, a_1 \oplus b_1, a_2 \oplus b_2, \dots$, and tries to decode these symbols using the peeling decoder.
Bob notifies Alice to stop when he has recovered all source symbols in $A \bigtriangleup B$. As we will soon show, the first symbol $a_0\oplus b_0$ in Rateless IBLT is decoded only after all source symbols are recovered. This is the indicator for Bob to terminate.

Linearity guarantees that $a_0\oplus b_0, a_1 \oplus b_1, a_2\oplus b_2, \dots$ is the coded symbol sequence for $A \bigtriangleup B$. 
Decodability guarantees that Bob can decode after receiving $O(|A \bigtriangleup B|)$ coded symbols and recover all source symbols in $A \bigtriangleup B$. 
Universality guarantees that Alice and Bob do not need any prior context to run the protocol.

If Alice regularly reconciles with multiple peers, she may cache coded symbols for $A$ to avoid recomputing them every session.  
Universality implies that Alice can reuse the same cached symbols across different peers.
Linearity implies that if she updates her set $A$, she can incrementally update the cached symbols by treating the updates $A \bigtriangleup A'$ as a set and subtracting its coded symbols from the cached ones for $A$.

We now discuss how we design an encoding algorithm that satisfies the three properties we set to achieve.

\subsubsection{Linearity \& Universality}
Our key observation is that to ensure linearity, it is sufficient to
define a consistent \emph{mapping rule}, which, given any source symbol $x \in \{0, 1\}^*$ and any index $i \ge 0$, deterministically decides whether $x$ should be mapped to the $i$-th coded symbol when encoding a set that contains $x$. This ensures that if $x \in A \cap B$, then it will be mapped to both $a_i$ and $b_i$ or neither; in either case, $x$ will not be reflected in $a_i \oplus b_i$. On the other hand, if $x \in A \bigtriangleup B$, and $x$ should be mapped to index $i$ according to the rule, then it will be mapped to \emph{exactly} one of $a_i$ or $b_i$, and therefore will be reflected in $a_i \oplus b_i$.
Since the mapping rule makes decisions based only on $x$ and $i$, the resulting encoding algorithm also satisfies universality. 

\subsubsection{Decodability}
\label{sec:decodability}

Whether the peeling decoder can recover all source symbols from a set of coded
symbols is fully determined by the mapping between the source and the coded
symbols.
Let $\rho(i)$ be the probability that a random source symbol maps to the $i$-th coded symbol, which we refer to as the \emph{mapping probability}. It is the key property that defines the behavior of a mapping rule. In the remainder of this subsection, we constrain $\rho(i)$ by examining two necessary conditions for peeling to succeed. Our key conclusion
is that in order for decodability to hold, $\rho(i)$ must be inversely proportional to $i$. This rejects most functions as candidates of $\rho(i)$
and leads us to a concrete instantiation of $\rho(i)$. We will design a concrete algorithm (mapping rule) that realizes the mapping probability in the next subsection, and mathematically prove that it satisfies decodability in \S~\ref{sec:analysis}.

\emph{First,} to kick-start the peeling decoder, there must be a coded symbol with exactly one source symbol mapped to it (a pure coded symbol). For a set $S$ and index $i$, the probability that this happens decreases quasi-exponentially in $\rho(i)|S|$.
This implies that $\rho(i)$ must decrease quickly with $i$. Otherwise, each of the first $O(|S|)$ coded symbols would have an exponentially small probability of being pure, and it would be likely that none of them is pure, violating decodability.

The following lemma shows that for this reason, the mapping probability $\rho(i)$ cannot decrease slower than $1/i^{1-\epsilon}$ for any positive $\epsilon$, i.e., \emph{almost} inversely proportional to $i$. 
We defer the proof to \S~\ref{sec:proofs}.

\begin{lemma}
\label{lemma1}For any $\epsilon > 0$, any mapping probability $\rho(i)$ such that $\rho(i) = \Omega\left(1/i^{1-\epsilon}\right)$, and any $\sigma > 0$,
if there exists at least one pure coded symbol within the first $m$ coded symbols for a random set $S$ with probability $\sigma$, then $m = \omega(|S|)$.
\end{lemma}

\emph{Second,} to recover all source symbols in a set $S$, we need at least $|S|$ non-empty coded symbols. This is because during peeling, each pure symbol (which must be non-empty) yields at most one source symbol.
Intuitively, $\rho(i)$ cannot decrease too fast with index $i$. Otherwise, the probability that a coded symbol is empty would quickly grow towards $1$ as $i$ increases. The first $O(|S|)$ coded symbols would not reliably contain at least $|S|$ non-empty symbols, violating decodability.

The following lemma shows that for this reason, the mapping probability $\rho(i)$ cannot decrease faster than $1/i$. We defer the proof to \S~\ref{sec:proofs}.

\begin{lemma}
\label{lemma2}For any mapping probability $\rho(i)$ such that $\rho(i)=o\left(1/i\right)$, and any $\sigma > 0$, if there exist at least $|S|$ non-empty coded symbols within the first $m$ coded symbols for a random set $S$ with probability $\sigma$, then $m = \omega(|S|)$.
\end{lemma}

The constraints above reject functions that decrease faster than $1/i$, as well as functions that decrease slower than $i^\epsilon/i$ for any $\epsilon>0$.
For simplicity, we ignore the degree of freedom stemming from the $i^\epsilon$ factor since for a sufficiently small $\epsilon$ and any practical $i$, it is very close to 1.
The remaining candidates for $\rho(i)$ are the ones in between, i.e., functions of order $1/i$.
We choose the simplest function in this class: 
\begin{equation}
    \label{eq:rhoi}\rho(i) = \frac{1}{1+\alpha i},
\end{equation}
where $\alpha > 0$ is a parameter. We shift the denominator by $1$ because $i$ starts at $0$.
In \S~\ref{sec:analysis}, we prove that this $\rho(i)$ achieves decodability with high efficiency: recovering a set $S$ only requires the first $1.35|S|$--$1.72|S|$ coded symbols on average.

We highlight two interesting properties of our $\rho(i)$. First, $\rho(0)=1$. This means that for any set, every source symbol is mapped to the first coded symbol. This coded symbol is only decoded after all source symbols are recovered. So, Bob can tell whether reconciliation has finished by checking if $a_0 \oplus b_0$ is decoded. Second, among the first $m$ indices, a source symbol is mapped to $\sum_{i=0}^{m-1}\rho(i)$ of them on average, or $O(\log m)$. It means that the \emph{density} of the mapping, which decides the computation cost of encoding and decoding, decreases quickly as $m$ increases. As we will show in~\S~\ref{sec:bench-compute}, the low density allows Rateless IBLT to achieve $2$--$2000\times$ higher throughput than PinSketch.

\subsection{Realizing the Mapping Probability}
\label{sec:mapping}

We now design an efficient deterministic algorithm for mapping a source symbol $s$ to coded symbols that achieves the mapping probability rule identified in the previous section.

Recall that for a random source symbol $s$, we want to make the probability that $s$ is mapped to the $i$-th coded symbol to be $\rho(i)$ in Eq.~\ref{eq:rhoi}. A simple strawman solution, for example, is to use a hash function that, given $s$, outputs a hash value uniformly distributed in $[0, 1)$. We then compare the hash value to $\rho(i)$, and decide to map $s$ to the $i$-th coded symbol if the hash value is smaller. Given a random $s$, because its hash value distributes uniformly, the mapping happens with probability $\rho(i)$. 

However, this approach has major issues. First, it requires comparing hash values and $\rho(i)$ for every pair of source symbol $s$ and index $i$. As mentioned in \S~\ref{sec:decodability}, the density of the mapping is $O(\log m)$ for the first $m$ coded symbols. In contrast, generating the $m$ coded symbols using this algorithm would require $m$ comparisons for each source symbol, 
significantly inflating the computation cost. Another issue is that we cannot use the same hash function when mapping $s$ to different indices $i$ and $j$. Otherwise, the mappings to them would not be independent: if $\rho(i) < \rho(j)$ and $s$ is mapped to $i$, then it will always be mapped to $j$. Using different, independent hash functions when mapping the same source symbol to different indices means that we also need to hash the symbol $m$ times.

We design an algorithm that maps each source symbol to the first $m$ coded symbols using only $O(\log m)$ computation.
The strawman solution is inefficient because we roll a dice (compare hash and $\rho(i)$) for every index $i$, even though we end up not mapping $s$ to the majority of them ($m-O(\log m)$ out of $m$), so reaching the next \emph{mapped} index takes many dice rolls ($m/O(\log m)$ on average).
Our key idea is to directly sample the distance (number of indices) to \emph{skip} before reaching the next mapped index.
We achieve it with constant cost per sample,
so we can jump from one mapped index straight to the next in constant time.

We describe our algorithm recursively. Suppose that, according to our algorithm, a source symbol $s$ has been mapped to the $i$-th coded symbol. We now wish to compute, in constant time, the \emph{next} index $j$ that $s$ is mapped to.
Let $G$ be the random variable such that $j-i = G$ for a random $s$, and 
let $P_g$ ($g \ge 1$) be the probability that $G = g$. In other words, $P_g$ is the probability that a random $s$ is not mapped to any of $i+1, i+2, \dots, i+g-1$, but is mapped to $i+g$, which are all independent events. So,
\begin{equation*}
    P_g =(1-\rho(i+1))(1-\rho(i+2))\dots(1-\rho(i+g-1))\rho(i+g).
\end{equation*}
Generating $j$ is then equivalent to sampling $g\leftarrow G$, whose distribution is described by $P_g$, and then computing $j=i+g$.

However, since there are $g$ (which can go to infinity) terms in $P_g$, it is
still unclear how to sample $G$ in constant time. The key observation is that
the cumulative mass function of $G$, denoted as $C(x)$, has a remarkably simple
form. In particular, 
\begin{equation}
    \label{eq:cx-closed}
    C(x) = \sum_{g=1}^x P_g=
1-\frac{\Gamma(i+1+\frac{1}{\alpha})\Gamma(x+i+1)}{\Gamma(i+1)\Gamma(x+i+1+\frac{1}{\alpha})}.
\end{equation}
We defer the step-by-step derivation to \S~\ref{appendix:mapping-calc}.

Let $C^{-1}(r)$ be the inverse of $C(x)$. The simple form of $C(x)$ allows us
to compute $C^{-1}(r)$ easily, which we will soon explain. To sample $G$, we
sample $r \leftarrow [0, 1)$ uniformly, and compute $g = \lceil
C^{-1}(r)\rceil$. To make the algorithm deterministic, $r$ may come from a
pseudorandom number generator seeded with the source symbol $s$.  The algorithm
outputs $i+g$ as the next index to which $s$ is mapped, updates $i \leftarrow
i+g$, and is ready to produce another index. Because every source symbol is
mapped to the first coded symbol (recall that $\rho(0)=1$), we start the
recursion with $i=0$.

Finally, we explain how to compute $C^{-1}(r)$. It is the most simple if we
set the parameter $\alpha$ in $\rho(i)$ to $0.5$. Plugging $\alpha = 0.5$ into
Eq.~\ref{eq:cx-closed}, we get $$C(x) = \frac{x(2i+x+3)}{(i+x+1)(i+x+2)}.$$ Its
inverse is 
\begin{align*}
    C^{-1}(r) &= \sqrt{\frac{(3+2i)^2-r}{4(1-r)}} -\frac{3+2i}{2}\\
              &\approx (1.5+i)((1-r)^{-\frac{1}{2}}-1).
\end{align*}
For a generic $\alpha$, we can use Stirling's approximation~\cite{stirling},
and get $$C(x) \approx 1-\left(\frac{i+1}{x+i+1}\right)^\frac{1}{\alpha}.$$
Consequently, $$C^{-1}(r) \approx (i+1)((1-r)^{-\alpha}-1).$$

In our final design, we set $\alpha = 0.5$. The main reason is that computing
$C^{-1}(r)$ when $\alpha=0.5$ only requires computing square roots, while it
otherwise involves raising $1-r$ to other non-integer powers. We
observe that the latter is significantly slower on older CPUs.
Meanwhile, as we will show in \S~\ref{sec:analysis}, setting $\alpha=0.5$
results in negligible extra communication compared to the optimal setting.

\subsection{Resistance to Malicious Workload}
\label{sec:hash-function}

In some applications, rogue users may inject items to Alice or Bob's sets. For example, in a distributed social media application where servers exchange posts, users can craft any post they like. This setting may create an ``adversarial workload,'' where the hash of the symbol representing the user's input does not distribute uniformly. 
If the user injects into Bob's set a source symbol that hashes to the same value as another source symbol that Alice has, then Bob will never be able to reconcile its set with Alice. This is because Bob will XOR the malicious symbol into the coded symbol stream it receives from Alice, but it will only cancel out the hash of Alice's colliding symbol from the \texttt{checksum} field, and will corrupt the \texttt{sum} field.

The literature on set reconciliation is aware of this issue, but typically does not specify the required properties from the hash function to mitigate it; most use hash functions with strong properties such as random oracles~\cite{multiparty}, which have long outputs (e.g., 256 bits). It is sufficient, however, to use a keyed hash function with uniform and shorter outputs (e.g., 64 bits). This allows Alice and Bob to coordinate a secret key and use it to choose a hash function from the family of keyed hashes. Although with short hashes, an attacker can computationally enumerate enough symbols to find a collision for an item that Alice has, the attacker does not know the key, i.e., the hash function that Alice and Bob use, so she cannot target a collision to one of Alice's symbols. This allows Rateless IBLT to minimize the size of a coded symbol and save bandwidth, particularly in applications where symbols are short and \texttt{checksums} account for much of the overhead. In practice, we use the SipHash~\cite{siphash} keyed hash function. A trade-off we make is that Alice has to compute the \texttt{checksums} separately for each key she uses, which increases her computation load. We believe this is a worthwhile trade-off as SipHash is very efficient, and we find in experiments (\S~\ref{sec:bench-compute}) that computing the hashes has negligible cost compared to computing \texttt{sums}, which are still universal. Also, we expect using different keys only in applications where malicious workload is a concern.

\section{Analysis}
\label{sec:analysis}
\label{analysis:decode}
\label{sec:deanalysis}

In this section, we use density evolution~\cite{density, density2} to analyze
the behavior of the peeling decoder when decoding coded symbols in Rateless
IBLTs. We mathematically prove that as the difference size $d$ goes to
infinity, the \emph{overhead} of Rateless IBLTs converges to $1.35$; i.e.,
reconciling $d$ differences requires only the first $1.35d$ coded symbols. We
then use Monte Carlo simulations to show the behavior for finite $d$. In
particular, we show that the overhead converges quickly, when $d$ is at the low
hundreds.

Density evolution is a standard technique for analyzing the iterative decoding
processes in error correcting codes based on sparse
graphs~\cite{density,density2}, and has been applied to IBLTs with simpler
mappings between source and coded symbols~\cite{irregular}. Its high-level idea
is to iteratively compute the probability that a random source symbol has
\emph{not} been recovered while simulating the peeling decoder statistically.
If this probability keeps decreasing towards $0$ as the peeling decoder runs
for more iterations, then decoding will succeed with probability converging to
$1$~\cite[\S~2]{density2}. The following theorem states our main conclusion.
We defer its proof to \S~\ref{sec:proofs}. 

\begin{theorem}
\label{detheorem}For a random set of $n$ source symbols, the probability that
    the peeling decoder successfully recovers the set using the first $\eta n$
    coded symbols (as defined in \S~\ref{sec:sequence}) tends to 1 as $n$ goes to
    infinity, provided that $\eta$ is any positive constant that satisfies
    \begin{equation}
    \label{eq:dethreshold}\forall q \in (0, 1]: e^{\frac{1}{\alpha}\text{Ei}\left(-\frac{q}{\alpha \eta}\right)} < q.
    \end{equation}
\end{theorem}

Recall that $\text{Ei}(\cdot)$ is the exponential integral
function;\footnote{$\text{Ei}(x) = -\int_{-x}^{\infty}\frac{e^{-t}dt}{t}$.}
$\alpha$ is the parameter in the mapping probability $\rho(i)=\frac{1}{1+\alpha
i}$ as discussed in \S~\ref{sec:sequence}. We stated Theorem~\ref{detheorem}
with respect to a generic set of source symbols and its corresponding coded
symbol sequence; in practice, the set is $A \bigtriangleup B$. The decoder
(Bob) knows the coded symbol sequence for $A \bigtriangleup B$ because he
subtracts the coded symbols for $B$ (generated locally) from those for $A$
(received from Alice) as defined in \S~\ref{sec:background}. 

Theorem~\ref{detheorem} implies that for any choice of parameter $\alpha$,
there exists a corresponding threshold $\eta^*$ which is the smallest $\eta$
that satisfies Eq.~\ref{eq:dethreshold}. Any $\eta > \eta^*$ also satisfies
Eq.~\ref{eq:dethreshold} because the left-hand side monotonically decreases
with respect to $\eta$. (Intuitively, this must be true as a larger $\eta$
means more coded symbols, which should be strictly beneficial for decoding.) As
long as Bob receives more than $\eta^*$ coded symbols per source symbol, he can
decode with high probability. In other words, $\eta^*$ is the communication
\emph{overhead} of Rateless IBLTs, i.e., the average number of coded symbols
required to recover each source symbol. $\eta^*$ is a function of $\alpha$. As
discussed in \S~\ref{sec:mapping}, we set $\alpha=0.5$ in our final design to
simplify the process of generating mappings according to $\rho(i)$. We solve
for $\eta^*$ when $\alpha=0.5$ and get the following result.

\begin{corollary}
\label{thm:ribltoverhead}The average overhead of Rateless IBLTs converges to
    $1.35$ as the difference size $d=|A \bigtriangleup B|$ goes to infinity.
\end{corollary}

\subsection{Monte Carlo Simulations}
\label{analysis:sim}

Theorem~\ref{detheorem} and Corollary~\ref{thm:ribltoverhead} from the density
evolution analysis state the behavior of Rateless IBLTs when the difference size
$d$ goes to infinity. To understand the behavior when $d$ is finite, we run
Monte Carlo simulations, and compare the results with the theorems.

Fig.~\ref{fig:parameter-alpha} shows the main results. It compares the overhead
predicted by Theorem~\ref{detheorem} and that observed in simulations. First,
notice that as the difference size increases, simulation results converge to
the analysis for all $\alpha$. How fast the results converge depends on
$\alpha$. For all $\alpha \le 0.55$, convergence happens quickly, and the
overhead observed in simulations stays within $10\%$ of the analysis even for
the smallest difference size we test. On the other hand, for $\alpha=0.95$,
simulation results are still $12\%$ higher than the analysis at the largest
difference size we test. Second, the figure shows that setting $\alpha=0.5$ is
close to optimal for the communication overhead. Setting $\alpha=0.5$ results
in $\eta^*=1.35$, while the optimal setting is $\alpha=0.64$ which results in
$\eta^*=1.31$, a difference of only $3\%$.

\begin{figure}
\centering
\includegraphics{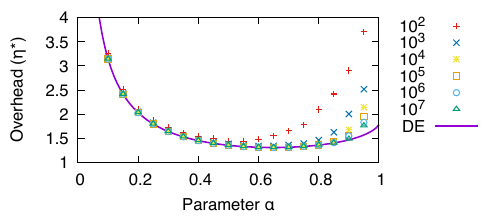}
    \caption[Relationship between the communication overhead $\eta^*$ and the
    parameter $\alpha$ in Rateless
    IBLTs.]{\label{fig:parameter-alpha}Relationship between the communication
    overhead $\eta^*$ and the parameter $\alpha$ in $\rho(i)$. ``DE'' shows 
    results from the density evolution analysis which assumes the difference
    size goes to infinity. Points show results from Monte Carlo simulations for various
    finite difference sizes. Each point is the average over 100
    runs.}
\end{figure}

Next, we focus on $\alpha=0.5$, the parameter we choose for our final design.
Fig.~\ref{fig:sim-overhead} shows the overhead as we vary the difference size
$d$. It peaks at $1.72$ when $d=4$ and then converges to $1.35$ as predicted by
Corollary~\ref{thm:ribltoverhead}. Convergence happens quickly: for all
$d>128$, the overhead is less than $1.40$.

\begin{figure}
\centering
\includegraphics{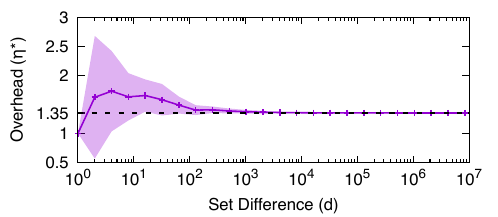}
    \caption[Overhead of Rateless IBLTs when reconciling differences of various
    sizes.]{\label{fig:sim-overhead}Overhead of Rateless IBLTs at varying
    difference sizes $d$. We run 100 simulations for each data point and report
    the average. The shaded area shows the standard deviation. The dashed line
    shows $1.35$, the overhead predicted by density evolution.}
\end{figure}

The density evolution analysis also predicts how decoding progresses as the
decoder receives more coded symbols. The fixed points of $q$ in
Eq.~\ref{eq:dethreshold} represent the expected fraction of source symbols that
the peeling decoder \emph{fails} to recover before stalling, as $d$ goes to
infinity. Fig.~\ref{fig:sim-process} compares this result with simulations (we
plot $1-q$, the fraction that the decoder \emph{can} recover) and they match
closely. There is a sharp increase in the fraction of recovered source symbols
towards the end, a behavior also seen in other codes that use the peeling
decoder, such as LT codes~\cite{lt}.

\begin{figure}
\centering
\includegraphics{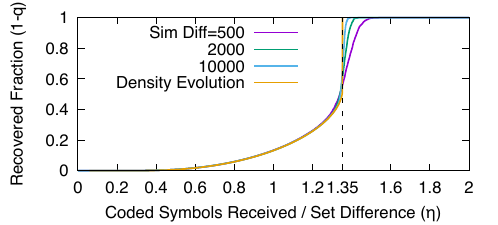}
    \caption[Fraction of recovered source symbols as the peeling decoder
    receives more coded symbols in Rateless IBLTs.]{\label{fig:sim-process}The
    fraction of recovered source symbols after receiving different number of
    coded symbols (normalized over the difference size $d$), as observed in
    simulations (average of 1000 runs), and as predicted by density evolution.
    Dashed line shows $1.35$, the overhead predicted by density evolution.}
\end{figure}

\section{Implementation}
\label{sec:implementation}

We implement Rateless IBLT as a library in 353 lines of Go code. The implementation is self-contained and does not use third-party code. In this section, we discuss some important optimizations in the implementation.

\smallskip
\noindent\textbf{Efficient incremental encoding.} A key feature of Rateless IBLT is that it allows Alice to generate and send coded symbols one by one until Bob can decode.
Suppose that Alice has generated coded symbols until index $i-1$, and now wishes to generate the $i$-th coded symbol. She needs to quickly find the source symbols that are mapped to it. A strawman solution is to store alongside each source symbol the next index it is mapped to, and scan all the source symbols to find the ones mapped to $i$. However, this takes $O(|A|)$ time. In our implementation, we store pointers to source symbols in a \emph{heap}. It implements a priority queue, where the priority is the index of the next coded symbol that a source symbol is mapped to. A smaller value indicates higher priority. This ensures that source symbols used for generating the next coded symbol are always at the head of the queue so that the encoder can access them efficiently without scanning all the source symbols.

\smallskip
\noindent\textbf{Variable-length encoding for \texttt{count}.} Recall that the \texttt{count} field stores the number of source symbols that are mapped to a coded symbol during encoding. The standard approach is to allocate a fixed number of bytes for it~\cite{difference,rate-compatible}, which inflates the size of each coded symbol by a constant amount. However, in Rateless IBLT, the value stored in \texttt{count} decreases with the index of the coded symbol according to a known pattern: the $i$-th coded symbol for a set $S$ is expected to have a \texttt{count} of $|S|\rho(i)$. This pattern allows us to aggressively compress the \texttt{count} field. Instead of storing the value itself, we can store the \emph{difference} of the actual value and the aforementioned expected value, which is a much smaller number. The node receiving the coded symbol can reconstruct the actual value of \texttt{count}, because it knows $N$ (transmitted with the $0$-th coded symbol) and $i$ (assuming a transport that preserves ordering). Instead of allocating a fixed number of bytes, we use variable-length quantity~\cite{vlq-encoding} to store the difference, which uses $\lceil\log_{128}{x}\rceil$ bytes to store any number $x$. Using our approach, the \texttt{count} field takes only $1.05$ bytes per coded symbol on average when encoding a set of $10^6$ items into $10^4$ coded symbols, keeping the resulting communication cost to a minimum.

\section{Evaluation}
\label{sec:bench}

We compare Rateless IBLT with state-of-the-art set reconciliation schemes, and demonstrate its low communication~(\S~\ref{sec:bench-communication}) and computation~(\S~\ref{sec:bench-compute}) costs across a wide range of workloads (set sizes, difference sizes, and item lengths). We then apply Rateless IBLT to synchronize the account states of Ethereum and demonstrate significant improvements over the production system on real workloads~(\S~\ref{sec:eval}). 

\smallskip
\noindent\textbf{Schemes compared.} We compare with regular IBLT~\cite{iblt,difference}, MET-IBLT~\cite{rate-compatible}, PinSketch~\cite{pinsketch}, and Merkle tries~\cite{merkle}. For Rateless IBLT, we use our implementation discussed in \S~\ref{sec:implementation}. For regular IBLT and MET-IBLT, we implement each scheme in Python. We use the recommended parameters~\cite[\S~6.1]{difference}\cite[\S\S~V-A, V-C]{rate-compatible}, and allocate 8 bytes for the \texttt{checksum} and the \texttt{count} fields, respectively. For PinSketch, we use Minisketch~\cite[\S~6]{erlay}, a state-of-the-art implementation~\cite{minisketch} written in C++ and deployed in Bitcoin. For Merkle tries, we use the implementation in Geth~\cite{geth}, the most popular client for Ethereum. 

\subsection{Communication Cost}
\label{sec:bench-communication}

We first measure the communication \emph{overhead}, defined as the amount of data transmitted during reconciliation divided by the size of set difference accounted in bytes. 
We test with set differences of $1$--$400$ items. Beyond $400$, the overhead of all schemes stays stable.
The set size is 1 million items (recall that it only affects Merkle trie's communication cost). Each item is 32 bytes, the size of a SHA256 hash, commonly used as keys in open-permission distributed systems~\cite{bitcoin,ipfs}. For Rateless IBLT and MET-IBLT, we generate coded symbols until decoding succeeds, repeat each experiment $100$ times, and then report the average overhead and the standard deviation. Regular IBLTs cannot be dynamically expanded, and tuning the number of coded symbols $m$ requires precise knowledge of the size of the set difference. Usually, this is achieved by sending an estimator before reconciliation~\cite{difference}, which incurs an extra communication cost of at least 15 KB according to the recommended setup~\cite[\S~V-C]{rate-compatible}. We report the overhead of regular IBLT with and without this extra cost. Also, unlike the other schemes, regular IBLTs may fail to decode probabilistically. We gradually increase the number of coded symbols $m$ until the decoding failure rate drops below $1/3\,000$. %

Fig.~\ref{fig:overhead} shows the overhead of all schemes except for Merkle trie, whose overhead is significantly higher than the rest at over $40$ across all difference sizes we test.
Rateless IBLT consistently achieves lower overhead compared to regular IBLT and MET-IBLT, especially when the set difference is small. For example, the overhead is $2$--$4\times$ lower when the set difference is less than $50$.
The improvement is more significant when considering the cost of the estimator for regular IBLTs. 
On the other hand, PinSketch consistently achieves an overhead of $1$, which is $37$--$60\%$ lower than Rateless IBLT. However, as we will soon show, Rateless IBLT incurs $2$--$2\,000\times$ less computation than PinSketch on both the encoder and the decoder. We believe that the extra communication cost is worthwhile in most applications for the significant reduction in computation~cost.

\begin{figure}
\centering
\includegraphics{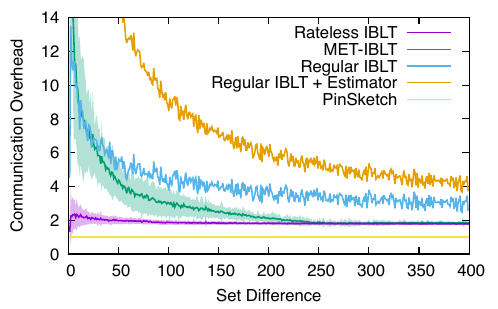}
\caption{\label{fig:overhead} Communication overhead of various schemes. Each item is 32 bytes. Shaded areas show the standard deviation for Rateless IBLT and MET-IBLT. Regular IBLT + Estimator shows the overhead of regular IBLT with an estimator for the size of set difference. We do not plot for Merkle Trie as its overhead is significantly higher (over 40) than the rest.}
\end{figure}

\smallskip
\noindent\textbf{Scalability of Rateless IBLT.} We quickly remark on how Rateless IBLT's communication cost scales to 
longer or shorter items.
Like other schemes based on sparse graphs, the \texttt{checksum} and \texttt{count} fields add a constant cost to each coded symbol. For Rateless IBLT, these two fields together occupy about $9$ bytes. Longer items will better amortize this fixed cost.
When reconciling shorter items, this fixed cost might become more prominent. However, it is possible to reduce the length of the \texttt{checksum} field if the differences are smaller, because there will be fewer opportunities for hash collisions. We found that hashes of $4$ bytes are enough to reliably reconcile differences of tens of thousands. It is also possible to remove the \texttt{count} field altogether; Bob can still recover the symmetric difference as the peeling decoder (\S~\ref{sec:background}) does not use this field. 

\subsection{Computation Cost}
\label{sec:bench-compute}

\begin{figure*}
    \centering
    \begin{minipage}[t]{.666\textwidth}
        \centering
        \begin{subfigure}[t]{0.5\textwidth}
         \centering
         \includegraphics[width=\textwidth]{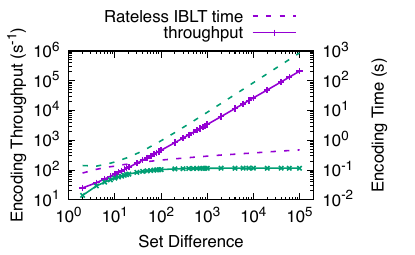}
         \caption{$N=1\,000\,000$}
         \label{fig:encoding-1m}
     \end{subfigure}%
     \hfill
     \begin{subfigure}[t]{0.5\textwidth}%
         \centering%
         \includegraphics[width=\textwidth]{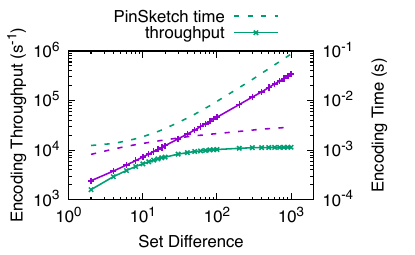}
         \caption{$N=10\,000$}
         \label{fig:encoding-10k}
     \end{subfigure}
     \captionsetup{width=.9\linewidth}
        \caption{Encoding throughput and time for sets of sizes $N=1\,000\,000$ and $N=10\,000$. Solid lines show the throughput (left Y-axis), and dashed lines show the encoding time (right Y-axis).}%
        \label{fig:encoding}
    \end{minipage}%
    \hfill
    \begin{minipage}[t]{0.333\textwidth}
        \centering
        \captionsetup{width=.9\linewidth}
        \includegraphics[width=\textwidth]{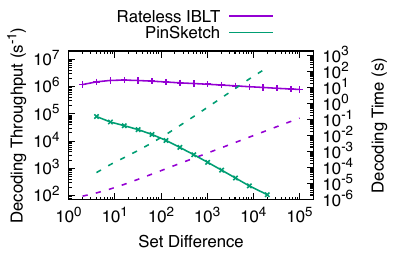}
        \caption{Decoding throughput and time. Solid lines show the throughput (left Y-axis), and dashed lines show the decoding time (right Y-axis).}
        \label{fig:decoding}
    \end{minipage}%
\end{figure*}

There are two potential computation bottlenecks in set reconciliation: encoding the sets into coded symbols, and decoding the coded symbols to recover the symmetric difference. Encoding happens at Alice, and both encoding and decoding happen at Bob. In this experiment, we measure the encoding and decoding throughput for sets of various sizes and differences. We focus on comparing with PinSketch. We fix the item size to $8$ bytes, because this is the maximum size that the PinSketch implementation supports. We do not compare with regular IBLT or MET-IBLT as we cannot find high-quality open source implementations, and they have similar complexity as Rateless IBLT.\footnote{The complexity is linear to the average number of coded symbols each source symbol is mapped to.%
This is $O(\log (m))$ for Rateless IBLT and MET-IBLT~\cite{rate-compatible}, and constant for regular IBLT, where $m$ is the number of coded symbols. However, the cost is amortized over the size of the set difference, which is $O(m)$. So, in all three IBLT-based schemes, the cost to encode for \emph{each} set difference decreases quickly as $m$ increases. } We will compare with Merkle trie in \S~\ref{sec:eval}.
We run the benchmarks on a server with two Intel Xeon E5-2697 v4 CPUs. Both Rateless IBLT and PinSketch are single-threaded, and we pin the executions to one CPU core using \texttt{cpuset(1)}.

\smallskip
\noindent\textbf{Encoding.} Fig.~\ref{fig:encoding} shows in solid lines the encoding throughput, defined as the difference size divided by the time it takes for the encoder to generate enough coded symbols for successful reconciliation. It indicates the number of items that can be reconciled per second with a compute-bound encoder. Rateless IBLT achieves $2$--$2000\times$ higher encoding throughput than PinSketch when reconciling differences of $2$--$10^5$ items. The significant gain is because the mapping between source and coded symbols is \emph{sparse} in Rateless IBLT, and the sparsity increases rapidly with $m$, so the average cost to generate a coded symbol decreases quickly. In comparison, generating a coded symbol in PinSketch always requires evaluating the entire characteristic polynomial, causing the throughput to converge to a constant. 

As the difference size increases, the encoding throughput of Rateless IBLT increases \emph{almost} linearly, enabling the encoder to scale to large differences. In Fig.~\ref{fig:encoding}, we plot in dashed lines the time it takes to finish encoding. As the difference size increases by $50\,000\times$, the encoding time of Rateless IBLT grows by less than $6\times$. Meanwhile, the encoding time of PinSketch grows by $5\,000\times$.

\smallskip
\noindent\textbf{Decoding.} Fig.~\ref{fig:decoding} shows the decoding throughput (solid lines) and time (dashed lines), defined similarly as in the encoding experiment. We do not make a distinction of the set size, because it does not affect the decoding complexity. (Recall that decoders operate on coded symbols of the symmetric difference only.) Rateless IBLT achieves $10$--$10^7\times$ higher decoding throughput than PinSketch. This is because decoding PinSketch is equivalent to interpolating polynomials~\cite{pinsketch}, which has $O(m^2)$ complexity~\cite{minisketch}, while decoding Rateless IBLT has only $O(m\log(m))$ complexity thanks to the sparse mapping between source and coded symbols. As the difference size grows by $50\,000\times$, the decoding throughput of Rateless IBLT drops by only $34\%$, allowing it to scale to large differences. For example, it takes Rateless IBLT $0.01$ second to decode $10^5$ differences. In contrast, it takes PinSketch more than a minute to decode $10^4$ differences.

\begin{figure}
    \begin{minipage}[t]{0.5\columnwidth}
        \centering
        \captionsetup{width=.9\linewidth}
        \includegraphics[width=\textwidth]{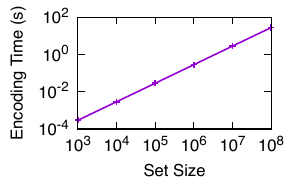}
        \caption{Encoding time of $1\,000$ differences and varying set size $N$.}
        \label{fig:encode-vary-set}
    \end{minipage}%
    \hfill
    \begin{minipage}[t]{0.5\columnwidth}
        \centering
        \captionsetup{width=.9\linewidth}
        \includegraphics[width=\textwidth]{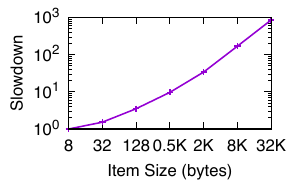}
        \caption{Slowdown when encoding items of different sizes.}
        \label{fig:field-size}
    \end{minipage}%
\end{figure}

\smallskip
\noindent\textbf{Scalability of Rateless IBLT.} We now show that Rateless IBLT preserves its computation efficiency when scaling to larger sets, larger differences, and longer items.

The set size $N$ affects encoding, but not decoding, because the decoder operates on coded symbols that represent the symmetric difference. The computation cost of encoding grows linearly with $N$, as each source symbol is mapped to the same number of coded symbols on average and thus adds the same amount of work. For example, in Fig.~\ref{fig:encoding}, the encoding time for $10^3$ differences is $2.9$ milliseconds when $N=10^4$, and $294$ milliseconds when $N=10^6$, a difference of $100\times$ that matches the change in $N$. Fig.~\ref{fig:encode-vary-set} shows the encoding time measured in experiments with the same configuration for a wider range of $N$.

The difference size $d$ affects both encoding and decoding. Recall that Rateless IBLT uses about $1.35d$ coded symbols to reconcile $d$ differences (\S~\ref{sec:analysis}). As $d$ increases, the encoder needs to generate more coded symbols. However, unlike PinSketch where the cost is linear in $d$, the cost of Rateless IBLT grows logarithmically. For example, in Fig.~\ref{fig:encoding-1m}, the encoding time grows by only $6\times$ as the set difference increases from $1$ to $10^5$. This is because the mapping from source to coded symbols is sparse: each source symbol is only mapped to an average of $O(\log d)$ coded symbols.
The same result applies to decoding.
For example, in Fig.~\ref{fig:decoding}, the decoding throughput drops by $2\times$ as the $d$ grows by $10^4\times$.

The item size $\ell$ affects both encoding and decoding because it decides the time it takes to compute the XOR of two symbols, which dominates the computation cost in Rateless IBLT. Fig.~\ref{fig:field-size} shows the relative slowdown when as $\ell$ grows from $8$ bytes to $32$ KB. Initially, the slowdown is sublinear (e.g., less than $4\times$ when $\ell$ grows by $16\times$ from $8$ to $128$ bytes) because the other costs that are independent of $\ell$ (e.g., generating the mappings) are better amortized. However, after $2$ KB, the slowdown becomes linear. This implies that the data rate at which the encoder can process source symbols, measured in bytes per second, stays constant. For example, when encoding for $d=1000$, the encoder can process source symbols at $124.8$ MB/s. The same analysis applies to decoding. In comparison, the encoding complexity of PinSketch increases linearly with $\ell$, and the decoding complexity increases quadratically~\cite{pinsketch,minisketch}.

\subsection{Application}
\label{sec:eval}

We now apply Rateless IBLT to a prominent application, the Ethereum blockchain.
Whenever a blockchain replica comes online, it must synchronize with others to get the latest \emph{ledger state} before it can validate new transactions or serve user queries.
The ledger state is a key-value table, where the keys are 20-byte wallet addresses, and the values are 72-byte account states such as its balance.
There are 230 million accounts as of January 4, 2024. Synchronizing the ledger state is equivalent to reconciling the set of all key-value pairs, a natural application of Rateless IBLT.

Ethereum (as well as most other blockchains) currently uses Merkle tries (\S~\ref{sec:related}) to synchronize ledger states between replicas. It applied a few optimizations: using a 16-ary trie instead of a binary one, and shortening sub-tries that have no branches. The protocol is called \emph{state heal} and has been deployed in Geth~\cite{geth}, the implementation which accounts for 84\% of all Ethereum replicas~\cite{geth-share}. Variants of Geth also power other major blockchains, such as Binance Smart Chain and Optimism. 

State heal retains the issues with Merkle tries despite the optimizations. To discover a differing key-value pair (leaf), replicas must visit and compare every internal node on the branch from the root to the differing leaf. This amplifies the communication, computation, and storage I/O costs by as much as the depth of the trie, i.e., $O(\log(N))$ for a set of $N$ key-value pairs.
In addition, replicas must descend the branch in lock steps, so the process takes $O(\log(N))$ round trips.
As a result, some Ethereum replicas have reported spending weeks on state heal, e.g.,~\cite{geth-stuck}.
In comparison, Rateless IBLT does not have these issues. Its communication and computation costs depend only on the size of the difference rather than the entire ledger state, and it requires no interactivity between replicas besides streaming coded symbols at line rate.

\smallskip
\noindent\textbf{Setup.}
We compare state heal with Rateless IBLT in synchronizing Ethereum ledger states.
We implement a prototype in 1,903 lines of Go code.
The prototype is able to load a snapshot of the ledger state from the disk, and synchronize with a peer over the network using either scheme. For state heal, we use the implementation~\cite{geth-impl} in Geth v1.13.10
without modification. For Rateless IBLT, we use our implementation discussed in \S~\ref{sec:implementation}. We wrap it with a simple network protocol where a replica requests synchronization by opening a TCP connection to the peer, and the peer streams coded symbols until the requesting replica closes the connection to signal successful decoding.  

To obtain workload for experiments, we extract snapshots of Ethereum ledger states as of blocks 18908312--18938312, corresponding to a 100-hour time span between December 31, 2023 and January 4, 2024. Each snapshot represents the 
ledger state when a block was just produced in the live Ethereum blockchain.\footnote{Ethereum produces a block every 12 seconds. Each block is a batch of transactions that update the ledger state.} For each experiment, we set up two replicas: Alice always loads the latest snapshot (block 18938312); Bob loads snapshots of different staleness and synchronizes with Alice.
This simulates the scenario where Bob goes offline at some point in time (depending on the snapshot he loads), wakes up when block 18938312 was just produced, and synchronizes with Alice to get the latest ledger state.
We run both replicas on a server with two Intel Xeon E5-2698 v4 CPUs running FreeBSD 14.0.
We use Dummynet~\cite{dummynet} to inject a 50 ms one-way propagation delay between the replicas and enforce different bandwidth caps of 10 to 100 Mbps.

\smallskip
\noindent\textbf{Results.}
We first vary the state snapshot that Bob loads and measure the completion time and the communication cost for Bob to synchronize with Alice.  
We fix the bandwidth to 20 Mbps. Fig.~\ref{fig:sync-distance} shows the results.
As Bob's state becomes more stale, more updates happen between his and the latest states, and the difference between the two grows linearly.
As a result, the completion time and the communication cost of both schemes
increase linearly.  Meanwhile, Rateless IBLT consistently achieves
$4.8$--$13.6\times$ lower completion time, and $4.4$--$8.6\times$ lower
communication cost compared to state heal. As discussed previously, state heal
has a much higher communication cost because it requires transmitting the
differing internal nodes of the Merkle trie in addition to the leaves. For
example, this amplifies the number of trie nodes transmitted by $3.6\times$
when Bob's state is 30 hours stale.  The higher communication cost leads to
proportionately longer completion time as the system is throughput-bound.

\begin{figure}
    \centering
    \begin{subfigure}[t]{\columnwidth}
        \centering
        \includegraphics{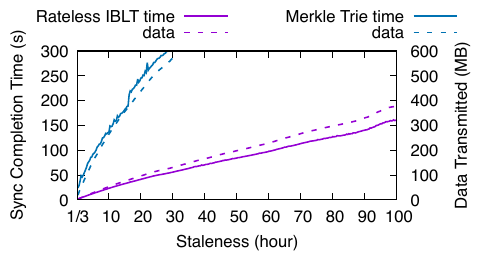}
        \caption{Staleness between 20 minutes and 100 hours.}
        \label{fig:sync-large}
    \end{subfigure}
    \begin{subfigure}[t]{\columnwidth}%
        \centering%
        \includegraphics{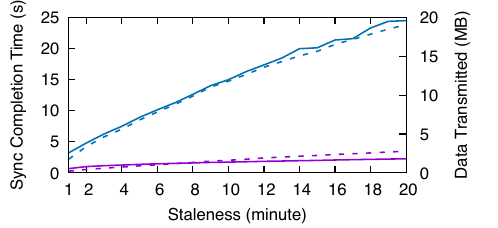}
        \caption{Staleness between 1 minute and 20 minutes.}
        \label{fig:sync-small}
    \end{subfigure}
    \caption{\label{fig:sync-distance} Completion time and communication cost when synchronizing Ethereum ledger states at different staleness over a network link with 50 ms of propagation delay and 20 Mbps of bandwidth. A staleness of $x$ hours means the state is $x$-hours old when synchronization starts.}
\end{figure}%

In our experiments, state heal requires at least 11 rounds of interactivity, as
Alice and Bob descend from the roots of their tries to the differing leaves in
lock steps. Rateless IBLT, in comparison, only requires half of a round because
Alice streams coded symbols without waiting for any feedback. This advantage is
the most obvious when reconciling a small difference, where the system would be
latency-bound.  For example, Rateless IBLT is $8.2\times$ faster than state
heal when Bob's ledger state is only 1 block (12 seconds) stale.  

We quickly highlight the impact of interactivity.  Fig.~\ref{fig:sync-trace}
shows traces of bandwidth usage when synchronizing one block worth of state
difference. For Rateless IBLT, the first coded symbol arrives at Bob in 1
round-trip time (RTT) after his TCP socket opens (0.5 RTT for TCP \texttt{ACK}
to reach Alice, and another 0.5 RTT for the first symbol to arrive). Subsequent
symbols arrive at line rate, as the peak at 1 RTT indicates. In comparison, for
state heal, Alice and Bob reach the bottom of their tries after 11 RTTs; before
that, they do not know the actual key-value pairs that differ, and the network
link stays almost idle.

\begin{figure}
\centering
\includegraphics{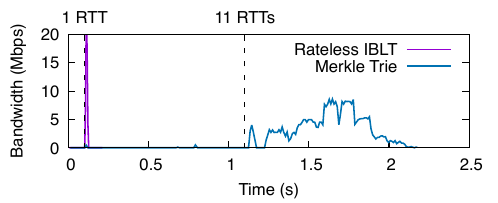}
    \caption[Bandwidth traces when synchronizing Ethereum ledger states]{\label{fig:sync-trace} Time series of bandwidth usage when synchronizing Ethereum ledger states that are 1 block (12 seconds) stale. The network link has 50 ms of propagation delay and 20 Mbps of bandwidth. Time starts when Bob sees the TCP socket open.}
\end{figure}

Finally, we demonstrate that Rateless IBLT consistently outperforms state heal across different network conditions. We fix Bob's snapshot to be 10 hours stale and vary the bandwidth cap. Fig.~\ref{fig:sync-bandwidth} shows the results. Rateless IBLT is $4.8\times$ faster than state heal at 10 Mbps, and the gain increases to $16\times$ at 100 Mbps. Notice that the completion time of state heal stays constant after 20 Mbps; it cannot utilize any extra bandwidth. We observe that state heal becomes \emph{compute-bound}: Bob cannot process the trie nodes he receives fast enough to saturate the network. The completion time does not change even if we remove the bandwidth cap. In contrast, Rateless IBLT is throughput-bound, as its completion time keeps decreasing with the increasing bandwidth. If we remove the bandwidth cap, Rateless IBLT takes 2.5 seconds to finish and can saturate a 170 Mbps link using one CPU core on each side.

\begin{figure}
\centering
\includegraphics{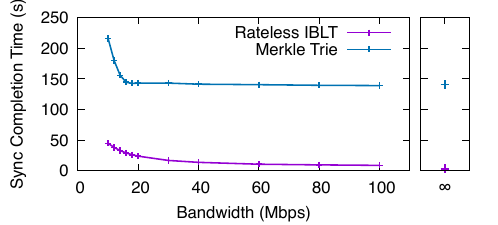}
    \caption[Completion time when synchronizing Ethereum states at different bandwidth]{\label{fig:sync-bandwidth} Completion time when synchronizing Ethereum ledger states that are 10 hours stale over a network link with 50 ms of propagation delay and different bandwidth.}
\end{figure}

Before ending, we quickly discuss a few other potential solutions and how Rateless IBLT compares with them. When Bob's state is consistent with some particular block, he may request Alice to compute and send the state \emph{delta} from his block to the latest block, which would be as efficient as an optimal set reconciliation scheme. However, this is often not the case when Bob needs synchronization, such as when he recovers from database corruption or he has downloaded inconsistent shards of state snapshots from multiple sources, routine when Geth replicas bootstrap~\cite{snapsync}. Rateless IBLT (and state heal) does not assume consistent states. Coded symbols in traditional set reconciliation schemes like regular IBLT are tailored to a fixed difference size (\S~\ref{sec:background}). Alice has to generate a new batch of coded symbols for each peer with a different state. This would add minutes to the latency for large sets like Ethereum, incur significant computation costs, and create denial-of-service vulnerabilities.\footnote{These issues also apply to the aforementioned state delta solution to a lesser degree, because Alice has to compute the state deltas on the fly.} In contrast, Rateless IBLT allows Alice to prepare a single stream of coded symbols that is efficient for all peers. Because of linearity (\S~\ref{sec:design}), Alice can incrementally update the coded symbols as her ledger state changes. For an average Ethereum block, it takes 11 ms to update 50 million coded symbols (7 GB) using one CPU core to reflect the state changes.

\section{Irregular Rateless IBLTs}
\label{sec:general}

When designing Rateless IBLTs (\S~\ref{sec:sequence}), the key task was to
define the mapping rule that decides whether a source symbol $x$ should be
mapped to the $i$-th coded symbol. Lemmas \ref{lemma1} and \ref{lemma2} showed
that the mapping rule cannot be uniform over coded symbols: the probability
that a random source symbol is mapped to the $i$-th coded symbol must decrease
with $i$. In other words, different coded symbols are statistically
inequivalent. On average, a coded symbol with a smaller index sees more source
symbols mapped to it than one with a larger index does.

However, the same is not true for source symbols in our design: every subset of
source symbols uses the same mapping probability $\rho(i) = \frac{1}{1+\alpha
i}$ with the same parameter $\alpha$. This leaves a degree of freedom which we
did not explore. We may divide source symbols into multiple subsets and use
different $\rho(i)$ (in particular, different $\alpha$) for each subset.
Similar techniques have successfully improved the communication costs of
regular IBLTs~\cite{irregular}. In this section, we apply this technique on
Rateless IBLTs and discuss the implications.

Concretely, we partition source symbols into $c$ mutually exclusive subsets,
where $c$ is a parameter. A random source symbol belongs to subset $j$ ($0 \le
j < c$) with probability $w_j$, which is another set of parameters.  For a
source symbol $x$, we choose the subset it belongs to based on its hash. For
example, $x$ belongs to subset $j$ if $\sum_{b=0}^{j-1} w_b \le
\mathtt{Hash}(x) < \sum_{b=0}^{j} w_b$ assuming the hash is uniformly
distributed in $[0, 1)$.  For each subset $j$, we define a parameter $\alpha_j$
and use mapping probability $\rho_j(i) = \frac{1}{1+\alpha_j i}$ when mapping
source symbols in this subset. In other words, we replace $\alpha$ with a
subset-specific $\alpha_j$ in the algorithm described in \S~\ref{sec:mapping}.
By convention, we call this generalized design Irregular Rateless IBLTs.
Rateless IBLTs as discussed prior to this section is a special case where $c=1$,
$w_0 = 1$, and $\alpha_0 = 0.5$.

As mentioned, the main benefit of Irregular Rateless IBLTs over Rateless IBLTs
is a lower communication cost. Unfortunately, the density evolution analysis
does not produce a closed-form result like the one in Theorem~\ref{detheorem}.
To find a good configuration of $c$, $w_j$, and $\alpha_j$ that minimizes the
overhead, we use brute force and try different values in simulations. To
limit the complexity of the search, we set the number of subsets $c$ to $3$ and
found the following optimal configuration
\begin{align*}
    &c = 3,\\
    &w_0, w_1, w_2 = 0.18,0.56,0.26,\\
    &\alpha_0, \alpha_1, \alpha_2 =0.11, 0.68, 0.82.
\end{align*}
As shown in Fig.~\ref{fig:overhead-irregular}, the resulting communication overhead converges to $1.10$, which is $19\%$ lower
than Rateless IBLTs (\S~\ref{sec:design}) and only $10\%$ above the
information-theoretic lower bound. Meanwhile, encoding and decoding are $1.88$
times slower than Rateless IBLTs. As mentioned in \S~\ref{sec:mapping}, the
main reason is that computing mappings when $\alpha \ne 0.5$ requires raising
numbers to arbitrary non-integer powers, while the case of $\alpha =
0.5$ only requires computing square roots, which is faster on modern hardware. We leave further optimizations of the parameters and the implementation
to future works.

\begin{figure}
\centering
\includegraphics{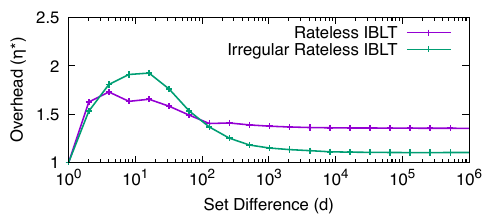}
    \caption[Communication overhead of Rateless IBLTs and Irregular Rateless IBLTs]{\label{fig:overhead-irregular} Communication overhead of Rateless IBLTs and Irregular Rateless IBLTs as the difference size changes. We run 100 simulations for each data point and report
    the average.}
\end{figure}

\section{Conclusion}
We designed, mathematically analyzed, and experimentally evaluated Rateless IBLT. To the best of our knowledge, Rateless IBLT is the first set reconciliation solution with universally low computation cost and near-optimal communication cost across workloads. The distinguishing feature is \emph{ratelessness}: it encodes any set into an infinitely long codeword, of which any prefix is capable of reconciling a proportional number of differences with another set. Ratelessness simplifies deployment as there is no parameter; reduces overhead as nodes can incrementally send longer prefixes without over- or under-committing resources to fixed-sized codewords; and naturally supports concurrent synchronization with multiple nodes. 
We mathematically proved its asymptotic efficiency and showed that the actual performance converges quickly with extensive simulations.
We implemented Rateless IBLT as a library and benchmarked its performance. Finally, we applied Rateless IBLT to a popular distributed application and demonstrated significant gains in state synchronization over the production system.

We point out a few interesting future directions: optimizing the parameters and the implementation of Irregular Rateless IBLTs; considering scenarios where Alice and Bob's sets change in the middle of reconciliation; and designing efficient solutions for reconciliation across more than two parties.

\section*{Acknowledgments}
We thank Francisco L{\'a}zaro for fruitful discussions. Lei Yang was supported
by a gift from the Ethereum Foundation. Yossi Gilad was partially supported by
the Alon Fellowship.

\smallskip\noindent
This work does not raise any ethical issues. Appendices are supporting material that has not been peer-reviewed.

\bibliographystyle{ACM-Reference-Format}

\bibliography{main}

\appendix
\section{Inflexibility of Regular IBLTs}
\label{appendix:iblt-musing}

We state and prove Theorems~\ref{thm:iblt-undersize} and~\ref{thm:iblt-oversize}, which show that the efficiency of regular IBLTs degrades exponentially fast when being used to reconcile more or fewer differences than parameterized for. We state the theorems with respect to generic sets of source symbols. When using IBLTs for set reconciliation (\S~\ref{sec:background}), the sets are $A\bigtriangleup B$.

\begin{theorem}\label{thm:iblt-undersize}
For a random set of $n$ source symbols and a corresponding regular IBLT with $m$ coded symbols, the probability that the peeling decoder can recover at least one source symbol decreases exponentially in $n/m$. 
\end{theorem}

\begin{proof}
For the peeling decoder to recover at least one source symbol, there must be at least one pure coded symbol at the beginning. Otherwise, peeling cannot start, and no source symbol can be recovered. We now calculate a lower bound on the probability $p_\text{nopure}$ that no pure coded symbol exists. Note that there is another parameter for regular IBLTs, $k$, which determines the number of coded symbols each source symbol is mapped to (\S~\ref{sec:background}). However, it can be shown that the probability that no pure coded symbol exists increases with $k$, so we set $k=1$ to get a lower bound.

We consider the equivalent problem: if we throw $n$ balls (source symbols) uniformly at random into $m$ bins (coded symbols), what is a lower bound on the probability that no bin ends up having exactly one ball? We compute the number of ways $f$ such that at least one bin has exactly one ball, which is the opposite of the event we are interested in. We set aside one of the $m$ bins which will get exactly one ball, and assign one of the $n$ balls to this bin. We then throw the remaining $n-1$ balls into the remaining $m-1$ bins freely. Notice that there are duplicates, so we get an upper bound
$$f \le mn(m-1)^{n-1}.$$
The total number of ways to throw $n$ balls into $m$ bins is $$g = m^n.$$ Each way of throwing is equally likely to happen, so the probability $p_\text{nopure}$ that no bin ends up with exactly one ball has a lower bound
\begin{align*}
p_\text{nopure} &= 1-f/g\\
&\ge 1-\frac{mn(m-1)^{n-1}}{m^n}\\
&= 1-\frac{n(m-1)^{n-1}}{m^{n-1}}\\
&=1-n\left(1-\frac{1}{m}\right)^{n-1}
\end{align*}

We are interested in the event where peeling can start, which is the opposite event. Its probability has an upper bound for $n/m > 1$
\begin{align*}
p_\text{haspure} &= 1-p_\text{nopure}\\
&\le n\left(1-\frac{1}{m}\right)^{n-1}\\
&\le n e^{-\frac{n-1}{m}}\\
&= o(1.5^{-\frac{n}{m}}).
\end{align*}
 
\end{proof}

The following theorem says that when dropping a fraction of a regular IBLT to reconcile a smaller number of differences $n$ with a constant overhead $\eta$ (the ratio between the number of used coded symbols and the number of source symbols $n$), the success probability decreases quickly as a larger fraction gets dropped. 

\begin{theorem}\label{thm:iblt-oversize}
Consider a random set of $n$ source symbols and a corresponding regular IBLT with $m$ coded symbols, where each source symbol is mapped to $k$ coded symbols.
The peeling decoder tries to recover all source symbols using the first $\eta n$ coded symbols. $k$ and $\eta$ are constants, and $\eta n\le m$. The probability that it succeeds decreases exponentially in $1-\eta n/m$.
\end{theorem}

\begin{proof}
For the peeling decoder to succeed, each of the source symbols must be mapped at least once to the first $\eta n$ coded symbols.
Because each source symbol is uniformly mapped to $k$ of the $m$ coded symbols, the probability that one is only mapped to the remaining $m-\eta n$ coded symbols that the decoder does not use (``missed'') is
\begin{align*}
p_\text{missone} &= \binom{m-\eta n}{k}/\binom{m}{k}\\
&= \frac{(m-\eta n)(m-\eta n-1)\dots(m-\eta n-k+1)}{m(m-1)\dots(m-k+1)}\\
&\approx \left(1-\frac{\eta n}{m}\right)^k.
\end{align*}
The last step approximates $\frac{m-\eta n-k+1}{m-k+1}$ with $\frac{m-\eta n}{m}$. This does not change the result qualitatively because $k$ is a constant. 

The probability that no source symbol is missed is
\begin{align*}
    p_\text{nomiss} &= (1-p_\text{missone})^n\\
    &\le e^{-np_\text{missone}}\\
    &= e^{-n\left(1-\frac{\eta n}{m}\right)^k}.
\end{align*}
\end{proof}

\section{Calculation of $P_g$ and $C(x)$}
\label{appendix:mapping-calc}
In this section, we calculate $P_g$ and $C(x)$ as defined in \S~\ref{sec:mapping}. 
\begin{align*}
P_g &=(1-\rho(i+1))(1-\rho(i+2))\dots(1-\rho(i+g-1))\rho(i+g)\\
&=\frac{1/\alpha}{i+g+1/\alpha}\prod_{n=1}^{g-1}\frac{i+n}{i+n+1/\alpha} \\
&=\frac{(i+1)_{g-1}}{\alpha(i+1+1/\alpha)_{g}}.
\end{align*}
Here, $(x)_n$ is the Pochhammer symbol.\footnote{$(x)_n = x (x+1)\dots(x+n-1)$.}

Before proceeding to calculate $C(x)$, we first prove a useful identity about quotients of Gamma functions,
\begin{equation*}
\frac{\Gamma(x)}{\Gamma(x+y)}\equiv\frac{1}{y-1}\left(\frac{\Gamma(x)}{\Gamma(x+y-1)}-\frac{\Gamma(x+1)}{\Gamma(x+y)}\right).
\end{equation*}
We start from the right-hand side,
\begin{align*}
&\frac{1}{y-1}\left(\frac{\Gamma(x)}{\Gamma(x+y-1)}-\frac{\Gamma(x+1)}{\Gamma(x+y)}\right) \\
=&\frac{1}{y-1}\left(\frac{(x+y-1)\Gamma(x)}{(x+y-1)\Gamma(x+y-1)}-\frac{\Gamma(x+1)}{\Gamma(x+y)}\right)\\
=&\frac{1}{y-1}\left(\frac{(y-1)\Gamma(x)+\Gamma(x+1)}{\Gamma(x+y)}-\frac{\Gamma(x+1)}{\Gamma(x+y)}\right)\\
=&\frac{\Gamma(x)}{\Gamma(x+y)}.
\end{align*}
The identity immediately implies the following
\begin{equation}
\label{eq:identity-gamma-series}\sum_{x=a}^b\frac{\Gamma(x)}{\Gamma(x+y)}\equiv\frac{1}{y-1}\left(\frac{\Gamma(a)}{\Gamma(a+y-1)}-\frac{\Gamma(b+1)}{\Gamma(b+y)}\right).
\end{equation}
We now calculate $C(x)$.
\begin{align*}
C(x) &= \sum_{g=1}^x P_g\\
&= \sum_{g=1}^x \frac{(i+1)_{g-1}}{\alpha(i+1+1/\alpha)_{g}}\\
&=\sum_{g=1}^x\frac{\Gamma(i+g)\Gamma(i+1+1/\alpha)}{\alpha\Gamma(i+1)\Gamma(i+1+g+1/\alpha)}\\
&=\frac{\Gamma(i+1+1/\alpha)}{\alpha\Gamma(i+1)}\sum_{g=1}^x\frac{\Gamma(i+g)}{\Gamma(i+g+1+1/\alpha)}\\
&=\frac{\Gamma(i+1+\frac{1}{\alpha})}{\Gamma(i+1)}\left(\frac{\Gamma(i+1)}{\Gamma(i+1+\frac{1}{\alpha})}-\frac{\Gamma(i+x+1)}{\Gamma(i+x+1+\frac{1}{\alpha})}\right)\\
&=1-\frac{\Gamma(i+1+\frac{1}{\alpha})\Gamma(x+i+1)}{\Gamma(i+1)\Gamma(x+i+1+\frac{1}{\alpha})}.
\end{align*}
The second last equality results from applying Eq.~\ref{eq:identity-gamma-series}.

\section{Deferred Proofs}
\label{sec:proofs}

\begin{lemma}[Restatement of Lemma~\ref{lemma1}]
For any $\epsilon > 0$, any mapping probability $\rho(i)$ such that $\rho(i) = \Omega\left(1/i^{1-\epsilon}\right)$, and any $\sigma > 0$,
if there exists at least one pure coded symbol within the first $m$ coded symbols for a random set $S$ with probability $\sigma$, then $m = \omega(|S|)$.
\end{lemma}

\begin{proof}
We need to show $\forall \eta > 0\, \exists |S|_0 > 0\, \forall |S| > |S|_0: m > \eta |S|$. Because $\rho(i) =\Omega\left(1/i^{1-\epsilon}\right)$, there exists $\delta > 0$ and $i_0 > 0$, such that
$\rho(i) \ge \delta/i^{1-\epsilon}$ for all $i > i_0$. Let $\rho_0$ be the smallest non-zero value among $\rho(i)$ for all $0 \le i \le i_0$. Let
$$|S|_0 = \max\left(\left(\frac{\eta^{1-\epsilon}}{\delta}\right)^{\frac{1}{\epsilon}},\frac{1}{\eta}\left(\frac{\delta}{\rho_0}\right)^{\frac{1}{1-\epsilon}}, |S|^*\right)$$
where $|S|^*$ is such that for all $|S| > |S|^*$, 
$$e\cdot\delta\eta^\epsilon |S|^{1+\epsilon} \exp\left(-{\frac{\delta|S|^\epsilon}{\eta^{1-\epsilon}}}\right) < \sigma.$$ Note that $\exp\left({\frac{\delta|S|^\epsilon}{\eta^{1-\epsilon}}}\right) = \omega\left(|S|^{1+\epsilon}\right)$, so such $|S|^*$ always exists.

For any $i \ge 0$, the $i$-th coded symbol is pure if and only if \emph{exactly} one source symbol is mapped to it, which happens with probability
\begin{align*}
P_i &= |S|\rho(i)(1-\rho(i))^{|S|-1}\\
&\le e\cdot|S|\rho(i)e^{-|S|\rho(i)}.
\end{align*}
The inequality comes from the fact that $(1-x)^y \le e^{-xy}$ for any $0 \le x \le 1$ and $y \ge 1$.

By the definition of $|S|_0$, for any $|S|>|S|_0$ and any $0 \le i \le \eta |S|$, either $\rho(i) =0$, or $\rho(i) \ge \frac{\delta}{(\eta |S|)^{1-\epsilon}}$ and $|S|\rho(i) > 1$. In either case,
$$P_i \le e \cdot \frac{\delta|S|^\epsilon}{\eta^{1-\epsilon}}\exp\left(-{\frac{\delta|S|^\epsilon}{\eta^{1-\epsilon}}}\right).$$

Recall that we want at least one pure symbol among the first $m$ coded symbols. Assume for contradiction that $m \le \eta |S|$. Then, failure happens with probability
\begin{align*}
P_\text{fail} &= \prod_{i=0}^{m-1}(1-P_i)\\
&\ge \prod_{i=0}^{\eta |S|-1}(1-P_i)\\
&\ge \left(1-e \cdot \frac{\delta|S|^\epsilon}{\eta^{1-\epsilon}}\exp\left(-{\frac{\delta|S|^\epsilon}{\eta^{1-\epsilon}}}\right)\right)^{\eta |S|}\\
&\ge 1-e\cdot\delta\eta^\epsilon |S|^{1+\epsilon} \exp\left(-{\frac{\delta|S|^\epsilon}{\eta^{1-\epsilon}}}\right)\\
&> 1-\sigma.
\end{align*}
\end{proof}

We remark that a stronger result which only requires $\rho(i) = \omega\left(\log i/i\right)$ can be shown with a very similar  proof, which we omit for simplicity and lack of practical implications. We may also consider a generalization of this lemma, by requiring there to be at least $k$ coded symbols with at most $k$ source symbols mapped to each, for every $k \le |S|$. (Lemma~\ref{lemma1} is the special case of $k=1$.) This may lead to an even tighter bound on $\rho(i)$, which we conjecture to be $\rho(i) = \omega(1/i)$.

\begin{lemma}[Restatement of Lemma~\ref{lemma2}]
For any mapping probability $\rho(i)$ such that $\rho(i)=o\left(1/i\right)$, and any $\sigma > 0$, if there exist at least $|S|$ non-empty coded symbols within the first $m$ coded symbols for a random set $S$ with probability $\sigma$, then $m = \omega(|S|)$.
\end{lemma}

\begin{proof}
We need to show that $\forall \eta > 0\, \exists |S|_0 > 0\, \forall |S| > |S|_0: m > \eta |S|$. First, note that for there to be $|S|$ non-empty symbols within the first $m$ coded symbols, $m$ cannot be smaller than  $|S|$, so the statement is trivially true for $0 < \eta < 1$. We now prove for the case of $\eta \ge 1$. 

For any $\eta \ge 1$, let $\delta = \frac{1}{4\eta}$. Because $\rho(i) = o(1/i)$, there must exist $i_0 > 0$ such that $ \rho(i) < \delta/i$ for all $i > i_0$.
Let $$|S|_0 = \max(2i_0, 4\eta^2(1-2\eta)\log(\sigma)).$$

For all $i \ge |S|/2$, the $i$-th coded symbol is non-empty with probability
\begin{align*}
P_i &= 1-\left(1-\rho\left(i\right)\right)^{|S|}\\
&< 1-\left(1-\frac{2\delta}{|S|}\right)^{|S|}\\
&\le 2\delta.
\end{align*}
The first inequality is because $\rho(i) < \delta/i$ for $i \ge |S|/2$, and the second inequality is because $(1-x)^y \ge 1-xy$ for any $x < 1$ and $y \ge 1$.

In order to get $|S|$ non-empty symbols among the first $m$ coded symbols, there must be at least $|S|/2$ non-empty symbols from index $i=|S|/2$ to index $i= m-1$.
To derive an upper bond on this probability,
we assume that each is non-empty with probability $2\delta$, which, as we just saw, is strictly an overestimate. By Hoeffding's inequality, the probability that there are at least $|S|/2$ non-empty symbols has an upper bound
$$P_\text{succ} < \exp\left(\left(|S|-2m\right)\left(2\delta-\frac{|S|}{2m-|S|}\right)^2\right)$$
when $m \le (\frac{1}{4\delta} + \frac{1}{2})|S|$, which is true for all $m \le \eta |S|$.

Assume $m \le \eta |S|$ for contradiction.
By the definition of $\delta$, the previous upper bound becomes
$$P_\text{succ} < \exp\left(\frac{|S|}{4\eta^2(1-2\eta)}\right).$$
The right hand side monotonically decreases with $|S|$. So, by the definition of $|S|_0$, for all $|S| > |S|_0$, 
$$P_\text{succ} < \exp\left(\frac{|S|_0}{4\eta^2(1-2\eta)}\right) \le\sigma.$$
\end{proof}

\begin{theorem}[Restatement of Theorem~\ref{detheorem}] For a random set of $n$ source symbols, the probability that the peeling decoder successfully recovers the set using the first $\eta n$ coded symbols (as defined in \S~\ref{sec:design}) goes to 1 as $n$ goes to infinity. Here, $\eta$ is any positive constant such that $$\forall q \in (0, 1]: e^{\frac{1}{\alpha}\text{Ei}\left(-\frac{q}{\alpha \eta}\right)} < q.$$
\end{theorem}

\begin{figure}
    \centering
    \begin{tikzpicture}[>=latex']
        \small
 
        \node[anchor=east] (setlabel) at (-0.6,-0.4) {Source symbols};
        \draw [draw=gray, thick, dashed, rounded corners] (-0.5,0.1) rectangle (3.2,-0.9);

        
        \node[draw, circle, text centered,fill=myred] (x0) at (0, -0.4) {$x_0$};
        \node[draw, circle, text centered,fill=myred] (x1) at (0.9, -0.4) {$x_1$};
        \node[draw, circle, text centered,fill=myred] (x2) at (1.8, -0.4) {$x_2$};
        \node[draw, circle, text centered,fill=myred] (x3) at (2.7, -0.4) {$x_3$};

        \node[anchor=east] (codelabel) at (-0.6,-2) {Coded symbols};
        \draw [draw=gray, thick, dashed, rounded corners] (-0.5,-1.5) rectangle (5.0,-2.5);
        \node[draw, rectangle, text centered, minimum size=1.8em,fill=mygreen] (a0) at (0, -2) {$a_0$};
        \node[draw, rectangle, text centered, minimum size=1.8em,fill=mygreen] (a1) at (0.9, -2) {$a_1$};
        \node[draw, rectangle, text centered, minimum size=1.8em,fill=mygreen] (a2) at (1.8, -2) {$a_2$};
        \node[draw, rectangle, text centered, minimum size=1.8em,fill=mygreen] (a3) at (2.7, -2) {$a_3$};
        \node[draw, rectangle, text centered, minimum size=1.8em,fill=mygreen] (a4) at (3.6, -2) {$a_4$};
        \node[draw, rectangle, text centered, minimum size=1.8em,fill=mygreen] (a5) at (4.5, -2) {$a_5$};

        \draw [] (x0) -- node[] {} (a0);
        \draw [] (x1) -- node[] {} (a0);
        \draw [] (x2) -- node[] {} (a0);
        \draw [] (x3) -- node[] {} (a0);
        \draw [] (x1) -- node[] {} (a1);
        \draw [] (x3) -- node[] {} (a1);
        \draw [] (x0) -- node[] {} (a2);
        \draw [] (x3) -- node[] {} (a2);
        \draw [] (x0) -- node[] {} (a3);
        \draw [] (x3) -- node[] {} (a3);
        \draw [] (x0) -- node[] {} (a4);
        \draw [] (x2) -- node[] {} (a4);
        \draw [] (x3) -- node[] {} (a5);

    \end{tikzpicture}
    \caption[Example of the bipartite graph representation of the mapping between source symbols and coded symbols]{\label{fig:iblt-graph} Example of the bipartite graph representation of a set of source symbols, $x_0, x_1, x_2, x_3$, and its first 6 coded symbols, $a_0, a_1, \dots, a_5$.}
\end{figure}
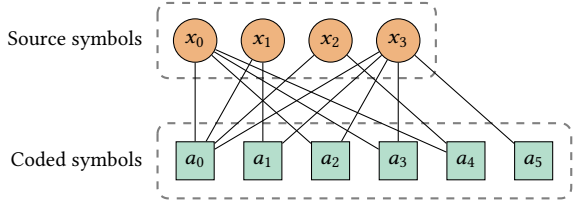

Before proving the theorem, we introduce the graph representation of a set of source symbols and the corresponding coded symbols. Imagine a bipartite graph where each source or coded symbol is a vertex in the graph, and there is an edge between a source and a coded symbol if and only if the former is mapped to the latter during encoding. Fig.~\ref{fig:iblt-graph} is an example.
We define the degree of a symbol as the number of neighbors it has in this bipartite graph, i.e., its degree as in graph theory. For example, in Fig.~\ref{fig:iblt-graph}, source symbol $x_0$ has degree 4, and coded symbol $a_1$ has degree 2.

We also define the degree of an \emph{edge} in the graph~\cite{irregular}. The \emph{source degree} of an edge is the degree of the source symbol it connects to, and its \emph{coded degree} is the degree of the coded symbol it connects to. For example, for the edge connecting $x_3$ and $a_3$ in Fig.~\ref{fig:iblt-graph}, its source degree is 5 because $x_3$ has degree 5, and its coded degree is 2 because $a_3$ has degree 2.

We remark that density evolution is a standard technique~\cite{density,density2} of analyzing codes that are based on random graphs, such as LT~\cite{lt} and LDPC~\cite{ldpc} codes. Our proof mostly follows these analysis, in particular, \cite[\S~2]{density2} and \cite[\S~III]{irregular}. However, the mapping probability $\rho(i)$ in Rateless IBLTs is a function whose parameter $i$ goes to infinity as the set size goes to infinity. This is a key challenge that we solve in our analysis, which enables us to get the closed-form expression in Theorem~\ref{detheorem}. 

\begin{proof}
Consider $n$ random source symbols and its first $m$ coded symbols.
Let $\Lambda$ be the random variable denoting the degree of a random source symbol. Let $\Lambda_u$ ($0 \le u \le m$) be the probability that $\Lambda$ takes value $u$.
Similarly, let $\Psi$ be the random variable denoting the degree of a random coded symbol. Let $\Psi_v$ ($0 \le v \le n$) be the probability that $\Psi$ takes value $v$.
Define the probability generating functions of $\Lambda$ and $\Psi$,
\begin{align*}
\Lambda(x) &= \sum_{u=0}^{m} \Lambda_u x^u, \\
   \Psi(x) &= \sum_{v=0}^{n} \Psi_v x^v.
\end{align*}

We also consider the degree of a random edge. Let $\lambda$ be the random variable denoting the source degree of a random edge. Let $\lambda_u$ ($0 \le u \le m$) be the probability that $\lambda$ takes value $u$.
It is the fraction of edges with source degree $u$ among all edges, i.e.,
\begin{align*}
\lambda_u &= \frac{\Lambda_u u}{\sum_{w=0}^{m}\Lambda_w w}\\
&=\frac{\Lambda_u u}{\mathbb{E}(\Lambda)}.
\end{align*}
Let $\lambda(x)$ be the generating function of $\lambda$, defined as
\begin{align*}
    \lambda(x) &= \sum_{u=0}^{m} \lambda_u x^{u-1}\\
    &= \frac{1}{\mathbb{E}(\Lambda)}\sum_{u=0}^{m} \Lambda_u u x^{u-1}\\
    &= \frac{\Lambda'(x)}{\mathbb{E}(\Lambda)}.
\end{align*}
Similarly, let $\varphi$ be the random variable denoting the coded degree of a random edge. Let $\varphi_v$ ($0 \le v \le n$) be the probability that $\varphi$ takes value $v$. It is the fraction of edges with coded degree $v$ among all edges, i.e.,
\begin{align*}
\varphi_v &= \frac{\Psi_v v}{\sum_{w=0}^{n}\Psi_w w}\\
&=\frac{\Psi_v v}{\mathbb{E}(\Psi)}.
\end{align*}
Let $\varphi(x)$ be the generating function of $\varphi$, defined as
\begin{align*}
    \varphi(x) &= \sum_{v=0}^{n} \varphi_v x^{v-1}\\
    &= \frac{1}{\mathbb{E}(\Psi)}\sum_{v=0}^{n} \Psi_v v x^{v-1}\\
    &= \frac{\Psi'(x)}{\mathbb{E}(\Psi)}.
\end{align*}

Let us now consider $\Psi(x)$.
Recall that each of the $n$ random source symbols is mapped to the $i$-th coded symbol independently with probability $\rho(i)$. The degree of the $i$-th coded symbol thus follows binomial distribution, which takes $v$ with probability $\binom{n}{v}\rho^v(i) (1-\rho(i))^{n-v}$. Because we are interested in a random coded symbol, its index $i$ takes $0, 1, \dots, m-1$ with equal probability $1/m$. By the law of total probability,
\begin{align*}
    \Psi_v = \frac{1}{m} \sum_{i=0}^{m-1} \binom{n}{v}\rho^v(i) (1-\rho(i))^{n-v}.
\end{align*}
Plugging it into the definition of $\Psi(x)$, we get
\begin{align*}
    \Psi(x) &= \sum_{v=0}^{n} \Psi_v x^v\\
    &=  \sum_{v=0}^{n} \frac{1}{m}\sum_{i=0}^{m-1} \binom{n}{v}\rho^v(i) (1-\rho(i))^{n-v} x^v \\
    &= \frac{1}{m}\sum_{i=0}^{m-1} \sum_{v=0}^{n} \binom{n}{v}(x\rho(i))^v (1-\rho(i))^{n-v}  \\
    &= \frac{1}{m}\sum_{i=0}^{m-1} (1-(1-x)\rho(i))^n.
\end{align*}
Here, the last step is because of the binomial theorem. Plugging it into the definition of $\varphi(x)$, we get
\begin{align*}
\varphi(x) &= \frac{n}{m\mathbb{E}(\Psi)} \sum_{i=0}^{m-1}\rho(i)(\rho(i)(x-1)+1)^{n-1}.
\end{align*}
By the handshaking lemma, which says the sum of the degree of all source symbols should equal the sum of the degree of all coded symbols, 
\begin{align*}
    m\mathbb{E}(\Psi) = m\sum_{v=0}^{n}\Psi_v v = n\sum_{u=0}^{m}\Lambda_u u = n \mathbb{E}(\Lambda).
\end{align*}
So, we can further simplify $\varphi(x)$ as
\begin{align*}
\varphi(x) &= \frac{1}{\mathbb{E}(\Lambda)} \sum_{i=0}^{m-1}\rho(i)(\rho(i)(x-1)+1)^{n-1}.
\end{align*}

Next, let us consider $\Lambda(x)$. For a random source symbol, it is mapped to the $i$-th coded symbol independently with probability $\rho(i)$. Its degree $\Lambda$ is thus the sum of independent Bernoulli random variables with success probabilities $\rho(0), \rho(1), \dots, \rho(m-1)$, which follows Poisson binomial distribution. By an extension~\cite[\S~5]{lecamextension} to Le Cam's theorem~\cite{lecam}, we can approximate this distribution with a Poisson distribution of rate $\sum_{i=0}^{m-1}\rho(i)$, i.e., $\mathbb{E}(\Lambda)$, with the total variation distance between the two distributions tending to zero as $m$ goes to infinity. That is,
\begin{align*}
    \sum_{u=0}^{\infty}\left|\Lambda_u - \frac{(\mathbb{E}(\Lambda))^u e^{-\mathbb{E}(\Lambda)}}{u!}\right| < \frac{2}{\mathbb{E}(\Lambda)}\sum_{i=0}^{m-1}{\rho^2(i)}.
\end{align*}
When $\rho(i)=\frac{1}{1+\alpha i}$ for any $\alpha > 0$, the right hand side of the inequality goes to zero as $m$ goes to infinity.

Recall that the probability generating function of a Poisson random variable with rate $\mathbb{E}(\Lambda)$ is
\begin{align*}
    \Lambda(x) = e^{\mathbb{E}(\Lambda)(x-1)}.
\end{align*}
Plugging it into the definition of $\lambda(x)$, we get
\begin{align*}
   \lambda(x) = e^{\mathbb{E}(\Lambda)(x-1)}.
\end{align*}

Let $q$ denote the probability that a randomly chosen edge connects to a source symbol that is \emph{not yet} recovered. As decoding progresses, $q$ is updated according to the following function~\cite{density2, irregular}
\begin{align*}f(q) &= \lambda(1-\varphi(1-q))\\
&= e^{-\mathbb{E}(\Lambda)\varphi(1-q)}\\
    &= e^{-\sum_{i=0}^{m-1}\rho(i)(1-q\rho(i))^{n-1}}.
\end{align*}

Let us consider $f(q)$ when the number of source symbols $n$ goes to infinity, and the ratio of coded and source symbols is fixed, i.e., $\eta = m/n$ where $\eta$ is a positive constant. 
Recall that $\rho(i) = \frac{1}{1+\alpha i}$. Notice that 
\begin{align*}
    e^{-\frac{nq}{\alpha i}}\le(1-q\rho(i))^{n-1} \le e^{-\frac{(n-1)q}{1+\alpha i}} 
\end{align*}
holds for all $n\ge 1$, $i \ge 0$, $\alpha > 0$, and $0 \le q \le 1$. We use this inequality and the squeeze theorem to calculate the limit of the exponent of $f(q)$ when $n$ goes to infinity.

We first calculate the lower bound. 
\begin{align*}
\lim_{n\rightarrow \infty}-\ln(f(q)) &=
    \lim_{n\rightarrow \infty} \sum_{i=0}^{\eta n-1}\rho(i)(1-q\rho(i))^{n-1}\\
    &\ge \lim_{n\rightarrow \infty}
    \sum_{i=0}^{\eta n-1}\rho(i)e^{-nq/(\alpha i)}\\
    &= \lim_{n\rightarrow \infty} \sum_{i=0}^{\eta n-1}\frac{1}{(1+\alpha i)e^{nq/(\alpha i)}}\\
    &= \lim_{n\rightarrow \infty} \frac{1}{n}\sum_{i=0}^{\eta n-1}\frac{1}{(\frac{1}{n}+\alpha \cdot\frac{i}{n})e^{\frac{q}{\alpha}\cdot \frac{n}{i}}}\\
    &= \int_{0}^{\eta}\frac{1}{\alpha x e^{\frac{q}{\alpha x}}}dx\\
    &= -\frac{1}{\alpha}\text{Ei}\left(-\frac{q}{\alpha \eta}\right).
\end{align*}
Here, $\text{Ei}(\cdot)$ is the exponential integral function.

We then calculate the upper bound.
\begin{align*}
\lim_{n\rightarrow \infty}-\ln(f(q)) &=
    \lim_{n\rightarrow \infty} \sum_{i=0}^{\eta n-1}\rho(i)(1-q\rho(i))^{n-1}\\
    &\le \lim_{n\rightarrow \infty}
    \sum_{i=0}^{\eta n-1}\rho(i)e^{-\frac{(n-1)q}{1+\alpha i}}\\
    &= \lim_{n\rightarrow \infty} \sum_{i=0}^{\eta n-1}\frac{1}{(1+\alpha i)e^{\frac{(n-1)q}{1+\alpha i}}}\\
    &= \lim_{n\rightarrow \infty} \frac{1}{(n-1)q}\sum_{i=0}^{\eta n-1}\frac{1}{\frac{1+\alpha i}{(n-1)q} e^{\frac{(n-1)q}{1+\alpha i}}}\\
    &= \frac{1}{\alpha}\int_{0}^{\alpha \eta / q}\frac{1}{x e^{1/x}}dx\\
    &= -\frac{1}{\alpha}\text{Ei}\left(-\frac{q}{\alpha \eta}\right).
\end{align*}
By the squeeze theorem,
\begin{align*}
\lim_{n\rightarrow \infty}-\ln(f(q))
    &= -\frac{1}{\alpha}\text{Ei}\left(-\frac{q}{\alpha \eta}\right).
\end{align*}
Plugging it into $f(q)$, we have
\begin{align*}
    \lim_{n\rightarrow \infty} f(q) &= e^{\frac{1}{\alpha}\text{Ei}\left(-\frac{q}{\alpha \eta}\right)}.
\end{align*}
By standard results~\cite[\S~2]{density2}\cite[\S~III.B]{irregular} of density evolution analysis, if
\begin{align*}
f(q) < q
\end{align*}
holds for all $q \in (0, 1]$, then the probability that all source symbols are recovered when the decoding process terminates tends to 1 as $n$ goes to infinity. Plugging in the closed-form result of $\lim_{n\rightarrow \infty} f(q)$, we get the condition $$e^{\frac{1}{\alpha}\text{Ei}\left(-\frac{q}{\alpha \eta}\right)} < q,$$ which should hold for all $ q \in (0, 1]$ for the success probability to converge to 1.
\end{proof}

We refer readers to the literature~\cite{density2} for a formal treatment on density evolution, in particular the result~\cite[\S~2.2]{density2} that $\forall q \in (0, 1] : f(q) < q$ is a sufficient condition for the success probability to converge to $1$, which we use directly in our proof. Here, we give some intuition. Recall that $q$ is the probability that a random edge in the bipartite graph connects to a source symbol that is \emph{not yet} recovered. 
Let $p$ be the probability that a random edge connects to a coded symbol that is \emph{not yet} decoded, i.e., has more than one neighbors that are not yet recovered.
Density evolution iteratively updates $q$ and $p$ by simulating the peeling decoder. For a random edge with source degree $u$, the source symbol it connects to is not yet recovered if none of the source symbol's other $u-1$ neighbors is decoded. This happens with probability $p^{u-1}$. Similarly, for a random edge with coded degree $v$, the coded symbol it connects to is not decoded if not all of the coded symbol's other $v-1$ neighbors are recovered. This happens with probability $1-(1-q)^{v-1}$.

Because $q$ and $p$ are the probabilities with regard to a random edge, we take the mean over the distributions of source and coded degrees of the edge, and the results are the new values of $q$ and $p$ after one iteration of peeling. In particular, in each iteration~\cite{irregular,density2},
\begin{align*}
p &\leftarrow \sum_v \varphi_v \left(1-(1-q)^{v-1}\right),\\
q &\leftarrow \sum_u \lambda_u p^{u-1}.
\end{align*}
By the definition of the generating functions of $\varphi$ and $\lambda$, the above equations can be written as
\begin{align*}
p &\leftarrow 1-\varphi(1-q),\\
q &\leftarrow \lambda(p).
\end{align*}
Combine the two equations, and we get 
\begin{align*}
q &\leftarrow \lambda(1-\varphi(1-q)).
\end{align*}
Notice that its right hand side is $f(q)$. Intuitively, by requiring $f(q) < q$ for all $q \in (0, 1]$, we make sure that the peeling decoder always makes progress, i.e., the non-recovery probability $q$ gets smaller, regardless of the current $q$. Conversely, if the inequality has a fixed point $q^*$ such that $f(q^*) = q^*$, then the decoder will stop making progress after recovering $(1-q^*)$-fraction of source symbols, implying a failure.

\end{document}